\documentclass[12pt]{article}

\usepackage[utf8]{inputenc}
\usepackage{amsmath,amssymb,amsthm,array,xcolor,multicol,verbatim,mathtools,microtype}
\allowdisplaybreaks

\usepackage{bbm}
\usepackage{bm}

\definecolor{myblue}{HTML}{0072B2}

\usepackage[hyperfootnotes=false]{hyperref}
\hypersetup{
	pdftitle={},
	pdfauthor={},
	pdfkeywords={},
	pdfcreator={LaTeX with hyperref package},
	pdfproducer={LaTeX},
	colorlinks=true,
	linkcolor=[rgb]{0,0,0.6},
	citecolor=[rgb]{0,0,0.6},
	urlcolor=[rgb]{0,0,0.6},
	pdfborder={0 0 0},
	pdfborderstyle={/S/U/W 0},
	pdfhighlight=/N,
	bookmarksnumbered=true,
	pdfdisplaydoctitle=true
}

\usepackage[normalem]{ulem}
\usepackage[pdftex]{graphicx}
\usepackage{fullpage}
\usepackage{enumitem}
    \setlist[itemize]{noitemsep,nolistsep}
    \setlist[enumerate,1]{noitemsep,nolistsep,label=(\arabic*)}
    \setlist[enumerate,2]{noitemsep,nolistsep,label=(\alph*)}
    \setlist[enumerate,3]{noitemsep,nolistsep,label=(\roman*)}

\usepackage{subcaption}

\usepackage[backend=biber,sorting=ynt,style=authoryear-comp,giveninits=true,maxbibnames=9,citetracker=true,maxcitenames=2,uniquename=false,uniquelist=false,natbib,isbn=false,doi=false,url=false,eprint=false]{biblatex}
\DeclareFieldFormat{pages}{#1}
\renewbibmacro{in:}{}

\usepackage[margin=1in]{geometry}
\usepackage{setspace}
\usepackage{titling}
\pretitle{\centering\large \bfseries}
\posttitle{\par\vskip 0.5em}
\preauthor{\centering\scshape\begin{tabular}[t]{c}}
	\postauthor{\end{tabular}\par}
\predate{\begin{center}}
	\postdate{\end{center}}

\usepackage{titlesec}
\titleformat{\section}{\centering\normalfont\scshape}{\thesection.}{1em}{}
\titleformat{\subsection}{\centering\normalfont\itshape}{\thesubsection.}{1em}{}
\titleformat{\subsubsection}{\normalfont}{\thesubsubsection.}{1em}{}

\titlespacing*{\subsection}{0em}{1em}{0.5em}
\titlespacing*{\subsubsection}{0em}{1em}{0.5em}

\usepackage{tikz}
\usetikzlibrary{calc,intersections,patterns}
\usepackage{pgfplots}
\usepgfplotslibrary{fillbetween}
\pgfplotsset{compat=1.17}

\theoremstyle{plain}
\newtheorem*{proposition*}{Proposition}

\newtheorem{theorem}{Theorem}
\newtheorem{proposition}{Proposition}
\newtheorem{lemma}{Lemma}

\newtheorem{step}{Step}

\theoremstyle{definition}
\newtheorem{definition}{Definition}
\newtheorem{assumption}{Assumption}
\newtheorem{remark}{Remark}
\newtheorem{example}{Example}

\newcommand{\continuation}{??}

\usepackage{etoolbox}
\expandafter\patchcmd\csname\string\proof\endcsname
{\normalparindent}{0pt }{}{}

\newcommand{\E}{\mathbb{E}}
\newcommand{\R}{\mathbb{R}}
\newcommand{\N}{\mathbb{N}}

\DeclareMathOperator*{\argmax}{arg\,max}
\DeclareMathOperator*{\argmin}{arg\,min}

\DeclareMathOperator{\col}{col}
\DeclareMathOperator{\rank}{rank}

\DeclareMathOperator{\sgn}{sgn}
\DeclareMathOperator{\Leb}{Leb}

\newcommand{\de}{\mathop{}\!\mathrm{d}}

\usepackage{accents}
\newcommand{\ubar}[1]{\underaccent{\bar}{#1}}

\usepackage[nameinlink,noabbrev]{cleveref}
\crefname{theorem}{Theorem}{Theorems}
\crefname{corollary}{Corollary}{Corollaries}
\crefname{lemma}{Lemma}{Lemmas}
\crefname{claim}{Claim}{Claims}
\crefname{assumption}{Assumption}{Assumptions}
\crefname{appendix}{Appendix}{Appendices}
\crefname{figure}{Figure}{Figures}

\newcommand{\headercref}[2]{\texorpdfstring{\Cref{#2}}{#1 \ref{#2}}}

\title{Eliciting Information from Opposing Factions\thanks{We thank Sarah Auster, Felix Bierbrauer, Daniel Bird, Martin Cripps, Francesc Dilm{\'e}, Piotr Dworczak, Tangren Feng, Duarte Gonçalves, Hans Peter Grüner, Yingni Guo, Philippe Jehiel, Johannes H\"orner, Jan Knoepfle, Nenad Kos, Daniel Kr\"ahmer, Patrick Lahr, Stephan Lauermann, Jiangtao Li, Moritz Meyer-ter-Vehn, Konrad Mierendorff, Benny Moldovanu, Georg N\"oldeke, Nikita Roketskiy, Larry Samuelson, Ludvig Sinander, Dezsö Szalay, Jonas von Wangenheim and seminar audiences from Berlin, Bonn, Caltech, Cologne (YEP), EC '22, ESSET (Gerzensee), LSE, Mannheim, Oxford, SMU (Singapore), SMYE, Stony Brook, Toulouse, UCL, and VSET. 
An extended abstract of an earlier version of this paper titled ``Mechanisms without Transfers for Fully Biased Agents'' appeared in the proceedings of EC '22.
}}
\author{
Deniz Kattwinkel\thanks{University College London, Department of Economics. Email: \textit{\href{mailto:d.kattwinkel@ucl.ac.uk}{d.kattwinkel@ucl.ac.uk}.}}
\and
Axel Niemeyer\thanks{Caltech, Division of the Humanities and Social Sciences. \textit{Email: \href{mailto:niemeyer@caltech.edu}{niemeyer@caltech.edu}}. Niemeyer was funded by the Deutsche Forschungsgemeinschaft (DFG, German Research Foundation) under Germany's Excellence Strategy -- EXC 2126/1 -- 390838866 when the first draft of this paper was written.}
\and
Justus Preusser\thanks{Bocconi University, Department of Economics and IGIER, \textit{\href{mailto:justus.preusser@unibocconi.it}{justus.preusser@unibocconi.it}}. Preusser acknowledges funding by the Deutsche Forschungsgemeinschaft (DFG, German Research Foundation) through CRC TR 224 (Project B04) and the European Research Council (HEUROPE 2022 ADG, GA No. 101055295 – InfoEcoScience). Preusser is grateful for the hospitality of Yale University.}
\and
Alexander Winter\thanks{University of Bonn, Department of Economics. Email: \textit{\href{mailto:awinter@uni-bonn.de}{awinter@uni-bonn.de}.} Winter funded was by the Deutsche Forschungsgemeinschaft (DFG, German Research Foundation) under Germany's Excellence Strategy -- EXC 2126/1 -- 390838866 and by ERC grant HWS - grant agreement No 101098319. }
}
\date{\today}

\begin{document}

\onehalfspacing

\maketitle

\vspace{-2em}
\begin{abstract}
\noindent

A principal makes a decision under risk and can consult two opposing factions. One side advocates for the decision and the other against, regardless of the state of the world. Can the principal elicit useful information from the two sides? Yes, under two conditions. First, the impact of one side's information on the principal's preferences must depend on what the other side knows (informational interdependence). Second, each side must hold some information that the other side cannot predict. In contrast to settings with transfers, richer correlation between the two sides' signals implies that fewer mechanisms are incentive-compatible. The optimal mechanism is the first-best decision rule under a virtual objective that handicaps self-serving reports. If the informational interdependence is rich enough, then useful information can be elicited even without commitment.

\end{abstract}

\section{Introduction}
A principal makes a decision under risk and can consult two opposing factions. 
Whichever way the principal leans, one side will insist the principal should take the decision and the other will insist the principal should not, regardless of what each side knows about the underlying state of the world. 
Can the principal nevertheless 
set up a scheme through which the two sides transmit useful information? 

This problem arises in a number of economic environments of interest.
For example, information about the efficient allocation of internal resources, such as budgets, is often dispersed across competing departments and must be aggregated by managers. 
Courts hear arguments from plaintiffs and defendants to reach a verdict. 
Similar situations arise in lobbying, where opposing interest groups advocate for and against a policy, or in the arbitration of labor disputes.

 Existing theories argue that the principal can elicit private information if, for example, the opposing parties are partially aligned with the principal, there are compromise options, the parties hold common information so that their reports can be checked against each other, or the principal has access to external information, verification tools, or transfer schemes.\footnote{On partial alignment, see, e.g., \citet{crawford1982strategic} and \citet{li2016competitive}; on compromise options, \citet{borgers2009efficient}, \citet{chakraborty2010persuasion}, and \citet{lipnowski2020cheap}; on common information, \citet{krishna2001model}, \citet{battaglini2002multiple}, \citet{ambrus2008multisender}, \citet{gerardi2009aggregation}, \citet{mylovanov2013decision}, and \citet{atakan2024cross}; on external information, \citet{kattwinkel2020allocation}, \citet{bloch2023selecting}, and \citet{kattwinkel2019costless}; on verifiable information, e.g., \citet{milgrom1986relying}, \citet{lipman1995robust}, \citet{dewatripont1999advocates}, and \citet{ben2014optimal,benporath2019evidence}; and on transfers, \citet{Cremer.1985,Cremer.1988}.}

Is information elicitation still possible in those bleak situations in which the principal cannot rely on partial alignment, compromise, shared information, or any additional instruments?
In this paper, we show that there is scope for information elicitation even when the two factions are diametrically opposed and the principal cannot rely on any of the other levers listed above.\footnote{In fact, our model allows for shared information, but we show that the principal cannot use it to cross-check the agents' reports (\Cref{proposition:spanning_order}).}
While real-world decision-makers need not always face the kinds of bleak situations studied here, the positive insights we obtain with the cards stacked against the principal carry over to any richer environment as well.

We find that communication is often possible if the impact  of one side's information on the principal's preferences non-trivially depends on the other side's information (informational interdependence) and if each side holds some information that the other side cannot predict. We further qualify exactly which information can be elicited and identify the optimal mechanism for information elicitation. The optimal mechanism may require the principal to commit to decisions that are wrong ex-post, i.e., given the principal's posterior beliefs. This kind of  commitment is written in law in various high-stakes situations, e.g., in criminal or competition law, and may also be achieved through reputation.\footnote{
For example, under U.S.\ federal law, prosecutors can enter into cooperation agreements in which they commit to move for a sentence reduction if the defendant provides substantial assistance, and, under such an agreement, the defendant's sentence can be reduced far below what their own admission of guilt would otherwise warrant (\citet[\S 5K1.1]{USSC2026}). 
Similarly, under EU competition law, a cartel member that is first to confess and cooperate receives full immunity from fines even if the investigation establishes the member as the cartel's instigator
(\citet[points 8, 12 and 13]{EC2006}).
} Nevertheless, we show that profitable communication is possible even without commitment when informational interdependencies between the two sides are sufficiently rich.

We derive these insights in the following model. A principal decides with which of two agents to side. The principal's payoff depends on the agents' noisy, private signals about a hidden state. Meanwhile, each agent has a state-independent preference for the principal's support. We allow for arbitrary correlation between the agents' signals. The principal designs and commits to a mechanism, i.e., a (stochastic) decision as a function of the reported signals. To explore the full range of when communication is possible, we begin by studying Bayesian incentive-compatible (BIC) mechanisms. Later, we remove the principal's commitment power.

First, we completely characterize BIC mechanisms in this setting with potentially correlated private information and no transfers (\Cref{lemma:BIC_characterization}). This setting is not amenable to any standard methods in mechanism design. We instead identify a connection between BIC mechanisms and correlated equilibria of an auxiliary two-player zero-sum game, which lets us leverage results on correlated equilibria in these games.\footnote{See \citet{rosenthal1974correlated}, \citet{forges1990correlated}, and \citet{viossat2006geometry}. \citet{kartik2013information,kartik2026information} draw an analogous connection between Bayes-Nash equilibria of two-player constant-sum Bayesian games with type-independent payoffs and correlated equilibria, with Downsian elections as their main application. We discuss the relationship in \Cref{sec:bic,sec:correlation}.} 

Second, we find that the set of BIC mechanisms shrinks if the correlation structure between the agents' private information is richer, in the sense of a stochastic order that we call the spanning order  (\Cref{proposition:spanning_order}). 
This finding starkly contrasts with the existing intuition from mechanism design with correlation (\citet{Cremer.1985,Cremer.1988}; \citet{mcafee1992correlated}) or multi-sender cheap talk (e.g., \citet{krishna2001model}; \citet{battaglini2002multiple}; \citet{ambrus2008multisender}; \citet{atakan2024cross}) according to which such correlation is useful for the principal. Indeed, under Cr\'emer and McLean's linear independence condition and with transfers,  the principal can elicit all private information without ceding any information rents to the agents.
In our setting without transfers, by contrast, 
the principal cannot elicit any private information under the same linear independence condition; conversely, the set of incentive compatible mechanisms is largest if the agents' signals about the state of the world are independent, i.e., if the agents have private private information in the sense of \citet{he2026private}.

Third, we show that the can principal nevertheless benefit from communicating with the agents. Assuming that the agents' signals are independent and the principal is ex-ante indifferent between the two decisions, the principal can do strictly better than flipping a coin if and only if the principal's payoff difference between the two decisions is not additively separable in the agents' signals (\Cref{proposition:profitability_independence}). Intuitively, non-additivity requires that the marginal impact of one agent's signal on the principal's preference must depend on the other agent's signal.

Fourth, we obtain the principal's optimal mechanism in the Myersonian tradition as the first-best mechanism under a virtual objective (\Cref{theorem:optimal_mech}). The virtual objective is derived from the non-additive part of the principal's true objective---in the sense of the previous paragraph---and we obtain it by connecting the principal's problem to a particular \emph{capacity-constrained optimal transport problem} in the sense of \citet{korman2015optimal} and \citet{korman2015dual}. We solve for the optimal mechanism in several illustrative examples.

Finally, we show that the principal can benefit from communicating with the agents even when the principal lacks commitment power (\Cref{theorem:no_commitment_profitable}). To do so, we strengthen the non-additivity condition mentioned above to a stronger non-separability condition. We also discuss how the principal can overcome their lack of commitment through a mediator.

Within the literature on mechanism design with correlated information but without transfers, closest to us are \citet{niemeyer2026optimal} and \citet{zhou2025peer}. Niemeyer and Preusser study an allocation problem without transfers in which each agent has information about the merits of allocating to the other agents. While Niemeyer and Preusser focus on this purely allocational channel of correlation by considering dominant-strategy incentive-compatible mechanisms, we consider Bayesian incentive-compatible mechanisms and can therefore explore all communication possibilities and how correlation affects the agents' reporting incentives. 
Zhang and Zhou study a principal who elicits information from two agents with correlated signals, but in a setup closer to the cheap-talk literature where agents have biased, state-dependent preferences.

\citet{antic2025assortative} study a principal who makes a separate binary decision for each of several agents with type-independent preferences. 
Communication benefits the principal for a reason related to our \Cref{proposition:profitability_independence}, as the principal's payoff in their model is not additively separable across agents.
The setting in which this result is obtained is very different from ours: in their model, each agent cares only about the decision that concerns themselves, whereas in our model, a single decision affects two parties with opposing interests.

We discuss further related work on zero-sum games and optimal transport alongside the relevant parts of our analysis.

\section{An Example}\label{sec:intro_example}
We begin by illustrating the insights of this paper by means of a simple example.\footnote{We thank Larry Samuelson for suggesting this example.}

A principal chooses between two options, $L$ and $R$. 
The viability $s_{\ell}$ of option $L$ is privately known to an interest group $\ell$, while the viability $s_{r}$ of $R$ is privately known to another interest group $r$.
For instance, the decision-maker may be a manager assigning an (unpleasant) task to one of two teams, or a politician consulting two opposing lobby groups.
Suppose $s_\ell$ and $s_r$ are independently and uniformly distributed on $[0, 1]$.\footnote{The optimal mechanism identified below can easily be adapted if $s_\ell$ and $s_r$ are independently distributed according to any common non-atomic distribution; see \Cref{sec:optimal_mechanism}.}
Group $\ell$ always advocates for $L$ and group $r$ always advocates for $R$. 
For simplicity, suppose each group receives a payoff of $1$ if their preferred option is chosen and a payoff of $-1$ otherwise.\footnote{Nothing would change if the payoff difference were any strictly positive number. The difference could also depend on each group's own signal.}
Likewise, suppose the principal receives a payoff of $1$ when choosing the more viable option (i.e., $L$ if and only if $s_\ell>s_r$), and $-1$ otherwise. 

Can the principal elicit any information from the two sides? Or will the two sides necessarily babble, binding the principal to their no-information payoff? 

The value of information in this example is $1$: if the principal does not consult the two groups, then the principal cannot do better than choosing one option at random, which yields the no-information payoff of $0$; if the principal had access to both groups' signals, then the principal would choose correctly and obtain a payoff of $1$. 

It turns out that the principal can achieve an expected payoff of $\tfrac 12$, i.e., recover 50\% of the total value of information. The unique optimal mechanism is:
\begin{quote}
\emph{each party $i=\ell,r$ separately makes a case $\hat s_i\in [0,1]$ to the principal; the principal sides with the party that makes the stronger case unless the gap $\vert\hat{s}_{\ell} - \hat{s}_{r}\vert$ between both cases exceeds one half, in which case the principal sides with the party that makes the weaker case.}
\end{quote}
The mechanism is Bayesian incentive compatible (BIC): each party gets their preferred option chosen with probability $\tfrac 12$ independently of their own report, given that the other party reveals their information; see \Cref{figure:horse_race}. Thus, it is
a best response for each party to make a truthful case, i.e., report $\hat{s}_{i} = s_{i}$. The principal then enjoys an expected payoff of $\tfrac12$ because the gap $\Delta=\vert s_\ell - s_r\vert$ between both cases is distributed with c.d.f.\ $\Delta(2-\Delta)$, hence the principal decides correctly in $75\%=\tfrac12(2-\tfrac12)$ of all instances. Optimality and uniqueness follow from the virtual-value approach underlying \Cref{theorem:optimal_mech} and are verified in \Cref{sec:optimal_mechanism} (\Cref{example:revisited}).

\begin{figure}[h]
\centering
\begin{tikzpicture}[scale=5]

    \tikzset{
        shaded region/.style={
            fill=gray!15,
            pattern color=gray!100
        }
    }

    \fill[shaded region]
        (0,0.5) -- (0,1) -- (0.5,1) -- cycle;

    \fill[shaded region]
        (0,0) -- (0.5,0) -- (1,0.5)
        -- (1,1) -- cycle;

    \draw[thick] (0,0.5) -- (0.5,1);
    \draw[thick] (0.5,0) -- (1,0.5);
    \draw[thick] (0,0) -- (1,1);

    \draw[thick] (0,0) rectangle (1,1);

    \node[below left] at (0,0) {$0$};
    \node[below] at (1,0) {$1$};
    \node[left] at (0,1) {$1$};

    \node[below=12pt] at (0.5,0) {$s_{\ell}$};
    \node[left=12pt] at (0,0.5) {$s_r$};

\draw[dashed, thick] (0.25,0) -- (0.25,0.25);
\draw[dashed, thick] (0.25,0.75) -- (0.25,1);

\draw[dashed,  thick] (0.5,0) -- (0.5,0.5);
\draw[dashed, thick] (0.75,0.25) -- (0.75,0.75);

\node[below, font=\scriptsize] at (0.25,0) {$1/4$};
\node[below, font=\scriptsize] at (0.5,0) {$1/2$};
\node[below, font=\scriptsize] at (0.75,0) {$3/4$};

\end{tikzpicture}
\caption{The optimal BIC mechanism in the illustrative example. The principal takes decision $L$ exactly for signal reports in the shaded area. To verify BIC, observe that from the perspective of agent $\ell$, the vertical line at any $s_{\ell}$ intersects the shaded area on a slice of height $1/2$, as depicted in the figure for $s_{\ell}\in\lbrace 1/4 , 1/2, 3/4\rbrace $. Thus, agent $\ell$ expects $L$ with probability $1/2$ no matter their report. Similarly for agent $r$.}
\label{figure:horse_race}
\end{figure}

Some useful information can also be elicited even if the principal cannot commit to a mechanism. Suppose each party sends one of two messages, $+1$ or $-1$, and the principal then takes the decision that is optimal given their posterior belief. The following is an equilibrium:
\begin{quote}
\emph{party $\ell$ says $+1$ if $s_\ell\in[0,\tfrac14)\cup[\tfrac34,1]$ and says $-1$ otherwise; party $r$ says $+1$ if $s_r\in[0,\tfrac18)\cup[\tfrac12,\tfrac78)$ and $-1$ otherwise; the principal chooses $L$ if the two messages coincide and $R$ if they differ.}
\end{quote}
For the two parties, the principal's strategy induces a game of matching pennies. In equilibrium, each party sends each message with probability $\tfrac12$, so either message of either party gets its preferred option chosen with probability $\tfrac12$. For the principal, it can be verified using Bayes' rule that the strategy above is a best response. The principal obtains an expected payoff of $\tfrac3{16}$ in this equilibrium.

The key insight of this paper is that information transmission is possible more broadly in situations such as the above, both with and without commitment.

\section{Model}
As in the example, a principal consults two biased agents, $\ell$ and $r$, in order to make a binary decision under risk, $L$ or $R$.\footnote{Each agent could be a representative of a group of individuals with common interests. See \Cref{sec:discussion:teams} for the details.} 
The principal's preferences over the decision now more generally depend on an unknown state of the world $\omega\in\Omega$. The principal's payoff from $L$ in state $\omega$ is given by $v(\omega)$ and is normalized to zero for $R$ without loss of generality. 
Each agent $i=\ell,r$ privately observes a signal $s_i\in S_i$ about the state. That is, $(\omega,s_\ell,s_r)$ are jointly distributed according to a probability distribution $\pi \in\Delta (\Omega\times S_{\ell}\times S_{r})$; abusing notation, we also use $\pi$ to denote the marginal distribution on $S=S_\ell\times S_r$.\footnote{Formally, each of $\Omega$, $S_\ell$, and $S_r$ is a Polish space equipped with its Borel sigma-algebra. Without loss of generality, the marginal distribution of $\pi$ on each of the three spaces has full support. The principal's payoff function $v\colon\Omega\to\R$ is integrable. We also make the mild technical assumption that the joint distribution $\pi$ is absolutely continuous with respect to the product of its marginals, $\pi_{\ell}\otimes \pi_{r}$, i.e., the distribution when the signals are independent with respective distributions $\pi_{\ell}$ and $\pi_{r}$. This is automatically satisfied when $S_\ell$ and $S_r$ are countable sets.\label{footnote:technical_assumptions}} 
We let $v(s_{\ell},s_{r}) = \E_\omega[v(\bm\omega)\mid s_{\ell}, s_{r}]$ denote the principal's expected payoff given signal profile $(s_{\ell},s_{r})$.\footnote{Bold symbols denote random variables.} 
While the principal's payoff from the decision is state-dependent, agent $\ell$ is biased in favor of $L$ and agent $r$ is biased in favor of $R$, no matter the state. That is, agent $i=\ell,r$ enjoys a payoff $\bar u_i(s_i)$ from their preferred decision and $\ubar{u}_i(s_i)$ from their less preferred decision, where $\bar u_i > \ubar{u}_i$.

The principal designs and commits to a mechanism. (We turn to the non-commitment case in \Cref{sec:no_commitment}.) The solution concept is Bayes-Nash equilibrium.\footnote{In the present model, only constant outcomes are implementable in dominant strategies, i.e., the principal would not be able to elicit any useful information from the agents under this more robust solution concept (see also \citet{niemeyer2026optimal}).} By the Revelation Principle, it is without loss of generality to focus on Bayesian incentive-compatible (BIC) direct mechanisms. A direct mechanism is a measurable function $x\colon S_{\ell} \times S_{r}\to [0, 1]$ that specifies the probability of decision $L$ for every reported signal profile. It is \emph{Bayesian incentive-compatible (BIC)} if each agent finds it optimal to truthfully report their signal given truthful reporting by the other agent, i.e., for $\pi_{\ell}$-almost all $s_{\ell}\in S_{\ell}$ and $\pi_{r}$-almost all $s_{r}\in S_{r}$,
\begin{alignat*}{2}
    \begin{aligned}
    s_{\ell} \in &\argmax_{s_{\ell}'\in S_{\ell}} \E_{s_{r}}\left[\bar{u}_{\ell}(s_{\ell}) x(s_{\ell}', \bm{s_{r}}) + \ubar{u}_\ell(s_{\ell})(1-x(s_{\ell}', \bm{s_{r}}))\mid s_{\ell}\right],
    \\
    s_{r} \in &\argmax_{s_{r}'\in S_{r}} \E_{s_{\ell}}\left[\ubar{u}_{r}(s_{r}) x(\bm{s_{\ell}}, s_{r}') + \bar{u}_{r}(s_{r}) (1 - x(\bm{s_{\ell}}, s_{r}')) \mid s_{r}\right].   
    \end{aligned}
\end{alignat*}
We refer to direct mechanisms simply as mechanisms.
The set of BIC mechanisms is denoted $\mathcal{X}(\pi)$.

The principal's problem is to maximize their expected payoff $\E_{s}[v(\bm{s_\ell},\bm{s_r})x(\bm{s_\ell},\bm{s_r})]$ subject to BIC, i.e., $x\in\mathcal{X}(\pi)$. A mechanism that solves the principal's problem is \textit{optimal}. 
A mechanism is \textit{profitable} if it is BIC and yields the principal a strictly greater payoff than choosing their ex-ante preferred option without consulting the agents, i.e., a payoff strictly greater than $\max\lbrace 0, \E_{s}[v(\bm{s})]\rbrace$.

\section{Incentive Compatibility}
\label{sec:bic}

We begin by characterizing the set of BIC mechanisms for arbitrary, potentially correlated signal distributions $\pi$.
Since our model features both correlated private information and non-transferable utility, standard methods and characterizations from the mechanism design literature cannot be applied.

Our key observation is that any BIC mechanism $x$ induces a two-player zero-sum game of which the joint distribution of signals $\pi$ is a correlated equilibrium.
Specifically, since $\bar u_{i}-\ubar{u}_i$ is always strictly positive, BIC is equivalent to: for $\pi_{\ell}$-almost all $s_{\ell}$ and $\pi_{r}$-almost all $s_{r}$,
\begin{alignat*}{2}
    \forall s_{\ell}' \in S_{\ell},\quad&  \E_{s_{r}}[x(s_{\ell}, \bm{s_{r}})\vert s_{\ell}] &&\geq  \E_{s_{r}}[x(s_{\ell}', \bm{s_{r}})\vert s_{\ell}], 
    \\
    \forall s_{r}' \in S_{r},\quad& \E_{s_{\ell}}[x(\bm{s_{\ell}}, s_{r}) \vert s_{r}] && \leq \E_{s_{\ell}}[x(\bm{s_{\ell}}, s_{r}') \vert s_{r}]. 
\end{alignat*}
That is, no agent can increase the interim probability of their preferred decision with their report conditional on their information.
Consider now the auxiliary zero-sum game in which agent $\ell$ is the maximizer and picks $s_{\ell}$, agent $r$ is the minimizer and picks $s_{r}$, and $\ell$'s payoff is $x(s_{\ell}, s_{r})$. With this interpretation, BIC says that the joint distribution $\pi$ of $(s_{\ell}, s_{r})$ forms a correlated equilibrium.
That is, conditional on being recommended action $s_i$, agent $i$ prefers to obey the recommendation rather than take any other action $s_i'$.
Note that, under the mechanism design interpretation, the joint distribution is an exogenous part of the environment while the mechanism is endogenous. In the auxiliary game, these roles are exactly flipped.

Based on this observation, existing results for correlated equilibria of two-player zero-sum games suggest that agents must be indifferent over many reports. Indeed, a result due to \citet[Proposition 3.8]{viossat2006geometry} implies that all obedience constraints of our correlated equilibrium $\pi$ must be binding when the ``action spaces'' $S_{\ell}$ and $S_{r}$ are finite.\footnote{Specifically, Viossat shows that two-player zero-sum games are pre-tight: in every correlated equilibrium, every obedience constraint against deviating to an action played with positive probability in some correlated equilibrium must hold with equality. Since the marginals of $\pi$ have full support by assumption, all obedience constraints of $\pi$ must hold with equality.
\citet{rosenthal1974correlated} and \citet{forges1990correlated} present related results on two-player zero-sum games with finite action spaces. Rosenthal shows that every correlated equilibrium yields every player the same utility as every Nash equilibrium. Forges notes that the conditional action distributions induced by any correlated equilibrium are minimax-optimal and hence, paired across players, constitute a Nash equilibrium with the same value.}
That is, for each agent $i$ and signal $s_i$, the interim-expected decision $\E_{s_{-i}}[x(s_{i}', \bm{s_{-i}})\mid s_{i}]$ must be constant in $i$'s report $s_i'$. 
With a few further steps, it can be shown that $\E_{s_{-i}}[x(s_{i}', \bm{s_{-i}})\mid s_{i}]$ must also be constant in $i$'s signal $s_{i}$. The following result generalizes these arguments to infinite signal spaces.

\begin{lemma}\label{lemma:BIC_characterization}
    If $x$ is a BIC mechanism, then for all $i\in\lbrace\ell, r\rbrace$, $\pi_{i}$-almost all signals $s_{i}\in S_{i}$, and $\pi_{i}$-almost all reports $s_{i}'\in S_{i}$,
    \begin{equation}\label{eq:interim_indifference_constancy}
        \E_{s_{-i}}[x(s_{i}', \bm{s_{-i}}) \mid s_{i}] = \E_{s_{-i}}[x(s_{i}, \bm{s_{-i}}) \mid s_{i}] = \E_{s}[x(\bm{s})].
    \end{equation}
    Conversely, if $x$ is a mechanism satisfying \eqref{eq:interim_indifference_constancy}, then there exists a BIC mechanism that agrees with $x$ at $\pi$-almost all signal profiles.
\end{lemma}
\begin{proof}
See \Cref{appendix:BIC_characterization}.
\end{proof}

Although it appears to be a stringent requirement that the interim-expected decision must be constant across agents, signals, and reports, we show in the subsequent sections that this requirement nevertheless leaves room for eliciting useful information from the agents. 
The fact that each agent is indifferent among all reports (so that truthful reporting is a weak best reply) is a consequence of our deliberately adversarial setting, in which the principal has no instruments besides the decision itself. Once the principal has access to an additional instrument, such as a small probability of detecting a lie, truth-telling incentives can be made strict.

\begin{remark}\label{remark:maxmin_value}
    If mechanism $x$ is BIC under signal distribution $\pi$, then the auxiliary zero-sum game induced by $x$ admits a value $\bar{x}$, which equals $\E_{s}[x(\bm{s})]$, i.e.
    \begin{equation*}
        \E_{s}[x(\bm{s})] 
        = 
        \bar{x}
        =
        \sup_{\sigma_{\ell}\in \Delta(S_{\ell})} 
        \inf_{\sigma_{r}\in\Delta(S_{r})} 
        \E_{s\sim\sigma_{\ell}\otimes\sigma_{r}}[x(\bm{s})]
        =
        \inf_{\sigma_{r}\in\Delta(S_{r})} 
        \sup_{\sigma_{\ell}\in \Delta(S_{\ell})}
        \E_{s\sim\sigma_{\ell}\otimes\sigma_{r}}[x(\bm{s})].
    \end{equation*}
    This is because the existence of a correlated equilibrium in a two-player zero-sum game implies the existence of a mixed Nash equilibrium, even with general action spaces; see \Cref{appendix:BIC_characterization}.
    In particular, since the value $\bar{x}$ is invariant to the signal distribution $\pi$, we conclude from \Cref{lemma:BIC_characterization} that each agent's interim-expectation is equal to $\bar{x}$ for all joint distributions $\pi$ under which $x$ is BIC.
\end{remark}

\citet{kartik2013information,kartik2026information} (KST) make an analogous observation for two-player constant-sum Bayesian games with type-independent payoffs: every Bayes-Nash equilibrium of such a game is a correlated equilibrium of the underlying complete-information game. They show that, in every Bayes-Nash equilibrium, all types of a player obtain the same interim payoff from almost every action played in equilibrium, and that this payoff is the same in every Bayes-Nash equilibrium. Due to their focus on Downsian elections and our focus on mechanism design and communication, however, the scope of both papers is very different. We discuss further connections and differences in \Cref{sec:correlation}.

In \Cref{sec:discussion:multiple_alternatives}, we discuss how to generalize the BIC characterization beyond binary decision problems.

\section{The Role of Correlation}
\label{sec:correlation}

In this section, we show that correlation between the agents' signals restricts the principal's ability to elicit useful information from the agents.
We define a stochastic order on the set of possible joint distributions $\pi$ that ranks the richness of the agents' interim beliefs.
We show that the set of BIC mechanisms shrinks in this order.

To simplify the exposition, we assume here that the sets of possible signals, $S_{\ell}$ and $S_{r}$, are finite. In \Cref{appendix:role_of_correlation}, we generalize the results to abstract signal spaces. With finite signal spaces, the signal distribution $\pi\in\Delta(S_\ell \times S_r)$  is a matrix in $\R_{+}^{\vert S_{\ell}\vert \times \vert S_{r}\vert}$ whose entries sum to one.\footnote{All joint distributions considered here are understood to have full-support marginals, consistent with the technical setup from \Cref{footnote:technical_assumptions}.} Given an agent $i$ and signal $s_{i}$, let $\pi_{i}(\cdot\mid s_{i})$ denote $i$'s posterior belief about $\bm{s_{-i}}$ conditional on $s_{i}$.

\begin{definition}
    Given joint distributions of signals $\pi, \pi' \in \Delta(S_{\ell}\times S_{r})$, we say that $\pi$ \emph{spans} $\pi'$ if for all $i\in\lbrace \ell, r\rbrace$ and $s_{i}'\in S_{i}$ the posterior belief $\pi'_{i}(\cdot\mid s_{i}')$ is a linear combination of the posterior beliefs $\lbrace \pi_{i}(\cdot \mid s_{i})\rbrace_{s_{i}\in S_{i}}$. 
\end{definition}
When $\pi$ spans $\pi'$, there is a sense in which the agents' signals are more informative about each other under $\pi$ than under $\pi'$.
We illustrate with several examples. 

\begin{example}
    Given marginals $\pi_{\ell}$ and $\pi_{r}$ of agent $\ell$'s and $r$'s signals, the product distribution $\pi_{\ell} \otimes\pi_{r}$ (i.e., $s_{\ell}$ and $s_{r}$ are independent unconditional on the state $\omega$) is spanned by every other joint distribution $\pi$ with the same marginals $\pi_{\ell}$ and $\pi_{r}$.\footnote{Under $\pi_{\ell} \otimes\pi_{r}$, agent $i$'s posterior about $\bm{s_{-i}}$ equals the prior, $\pi_{-i}$. Since $\pi$ has marginal $\pi_{-i}$ on $S_{-i}$, for all $s_{-i}$ we have $\pi_{-i}(s_{-i}) = \sum_{s_{i}\in S_{i}} \pi_{i}(s_{-i}\mid s_{i}) \pi_{i}(s_{i})$. Thus, the posterior is indeed a linear combination of $\lbrace \pi_{i}(\cdot \mid s_{i})\rbrace_{s_{i}\in S_{i}}$.}
    In the terminology of \citet{he2026private}, the product distribution $\pi_{\ell} \otimes\pi_{r}$ corresponds to a private private information structure: the agents' signals are uninformative about each other but may nevertheless be informative about the state $\omega$ that is payoff-relevant for the principal.
\end{example}

\begin{example}\label{example:boundedly_complete}
    If $\pi$ has full column- and row-rank, then $\pi$ spans every other joint distribution.
    This is exactly the condition introduced by Cr{\'e}mer and McLean (see Assumption 4 in their \citeyear{Cremer.1985} paper and Theorem 1 in their \citeyear{Cremer.1988} paper).
    For example, suppose for each agent the possible signals are $-1$ and $1$, and the joint distribution $\pi$ is given by the following matrix:
    \begin{center}
    \begin{tabular}{c|cc|}
               \multicolumn{1}{c}{} & \multicolumn{1}{c}{$-1$}  & \multicolumn{1}{c}{$1$} \\\cline{2-3}
               $-1$ & $0.25-\varepsilon$ & $0.25+\varepsilon$\\$1$ & $0.25+\varepsilon$ & $0.25-\varepsilon$ \\\cline{2-3}
    \end{tabular}
    \end{center}
    If $\varepsilon \neq 0$, the matrix has full rank. If $\varepsilon = 0$, the signals are independent.
\end{example}

Finally, we consider two intermediate cases in which the signals are not necessarily independent but the joint distribution also does not necessarily have full rank. 
\begin{example}
    Let $\Omega$ be finite.
    Suppose that $s_{\ell}$ and $s_{r}$ are independent \emph{conditional} on the state $\omega$.
    In this case, the rank of $\pi$ is at most $\vert\Omega\vert$.\footnote{Agent $i$'s posterior conditional on $s_{i}$ is given by $\pi_{i}(s_{-i} \mid s_{i}) = \sum_{\omega\in\Omega} \pi(s_{-i}\mid\omega) \pi(s_{i}\mid\omega) \pi(\omega) / \pi_{i}(s_{i})$. In particular, $\pi_{i}(\cdot\mid s_{i})$ is a linear combination of the vectors $\lbrace\pi(\cdot\mid\omega)\rbrace_{\omega\in\Omega}$, of which there are at most $\vert\Omega\vert$.}
    Therefore, if there are fewer states than signals of each agent, then $\pi$ does not have full rank, but $s_{\ell}$ and $s_{r}$ need not be independent either. 
\end{example}

\begin{example}
\label{example:two_components}
    Suppose each agent's signal $s_i$ consists of two components: $s_{i} = (s_{i, 1}, s_{i, 2})$. Suppose the first components $s_{\ell, 1}$ and $s_{r, 1}$ are correlated across agents, but $(s_{\ell, 1}, s_{r, 1})$, $s_{\ell, 2}$, and $s_{r, 2}$ are mutually independent.
    Then, $\pi$ does not have full rank, but $s_{\ell}$ and $s_{r}$ are also not independent. 
    Every mechanism that depends only on the latter components and that would be BIC if these latter components were the agents' only signals is also BIC under $\pi$.
\end{example}

Our next result connects the spanning order on joint distributions $\pi$ to the set $\mathcal{X}(\pi)$ of BIC mechanisms.
We find that correlation between the agents' signals shrinks the set of BIC mechanisms.
\begin{proposition}\label{proposition:spanning_order}
    For all $\pi, \pi'\in \Delta(S_{\ell}\times S_{r})$, if $\pi$ spans $\pi'$, then $\mathcal{X}(\pi) \subseteq \mathcal{X}(\pi')$.
    Consequently, if a mechanism $x$ is BIC for some $\pi\in\Delta(S_{\ell}\times S_{r})$, then $x$ is also BIC for the product distribution $\pi_{\ell}\otimes \pi_{r}$ with the same marginals.
    Moreover, if $\pi\in\Delta(S_{\ell}\times S_{r})$ has full column-rank or full row-rank, then all $x\in\mathcal{X}(\pi)$ are constant; the converse holds if $\pi$ has full support. 
\end{proposition}
\begin{proof}
See \Cref{appendix:role_of_correlation}.
\end{proof}

\citet{Cremer.1985,Cremer.1988} show in settings with transfers that full-rank correlation makes it possible to implement any decision rule while paying zero information rents.
We show that, in our setting, the same full-rank condition implies that only mechanisms that ignore the agents' reports are BIC.\footnote{The fact that under full-rank correlation all BIC mechanisms are constant is a consequence of our deliberately adversarial setting in which the principal has no instruments besides the decision itself. When the principal has additional instruments for providing incentives, such as when they have a small probability of detecting a lie, reporting incentives in any non-constant BIC mechanism under the product distribution $\pi_{\ell}\otimes\pi_{r}$ can be made strict. Each agent's loss from misreporting is continuous in the joint distribution $\pi$ with respect to total-variation distance. Thus, with additional instruments, non-constant BIC mechanisms exist even for full-rank joint distributions that are close to $\pi_{\ell}\otimes\pi_{r}$.}
That is, meaningful communication is impossible if agents are too informed about each other's information. By contrast, when signals are independent, the set of implementable decision rules is largest.
Thus, the beneficial role of correlation in mechanism design with transfers, as identified by \citeauthor{Cremer.1985}, is exactly reversed in our setting without transfers.

We illustrate this difference for the binary signals from \Cref{example:boundedly_complete}. For $\varepsilon = 0$, the signals $s_{\ell}$ and $s_{r}$ are independent. In this case, the following is a BIC mechanism: take decision $L$ if and only if the reported signals match, $x(s) = \mathbbm{1} (s_{\ell} = s_{r})$ for all $s$, since then both agents expect $L$ to be taken with probability $1/2$ no matter their true signal and no matter their report. If $\varepsilon > 0$, this mechanism is no longer BIC. Consider agent $\ell$ with signal realization $s_{\ell} = 1$. Then, agent $\ell$ has a posterior belief of $1/2 + 2\varepsilon$ that agent $r$'s signal is $s_{r} = -1$. Since agent $\ell$ prefers decision $L$, they misreport $s_{\ell}' = -1$. Thus, $x$ is not BIC. Indeed, \Cref{proposition:spanning_order} implies that only constant mechanisms are BIC when $\varepsilon > 0$, even if signals are perfectly correlated.

While full-rank correlation rules out any meaningful communication in our setting, let us explain why we do not regard this case as the economically natural one. Full-rank correlation implies that an agent’s signal is fully determined by their belief about the other agent’s signal (see also \citet{NE2004}). This property holds trivially for correlated binary signals and, more generally, in standard economic models with correlation where agents observe noisy signals of a common one-dimensional state. But payoff-relevant states in any realistic decision problem depend on a multitude of factors and are typically high-dimensional objects. If each agent learns about some of these factors from sources unrelated to anything else either agent observes, then an agent’s belief cannot determine the agent’s signal, and there cannot be full-rank correlation, as in \Cref{example:two_components}.

Since we posted the first version of our paper (\citet{kattwinkel2022mechanisms}), an analogue of the full-rank impossibility part of our \Cref{proposition:spanning_order} has been obtained by \citet{kartik2026information} for two-player Bayesian zero-sum games with state-independent preferences.\footnote{The other parts of \Cref{proposition:spanning_order}, the spanning-order inclusion and the converse of the full-rank statement, have no counterpart in KST. Their 2013 draft (\citet[Lemma B.2]{kartik2013information}) contains what is essentially a special case of their result for normally distributed signals.} The context in which KST apply this result is different from ours: they mainly study Downsian elections, whereas we study mechanism design and communication. KST consider infinite signal spaces and more general zero-sum preferences than we did in our first version. We supply the generalizations to infinite signals in \Cref{appendix:role_of_correlation}, as some of our later results now assume continuous signal spaces, and we explain how to handle more general preferences in \Cref{sec:discussion:multiple_alternatives}. Whereas KST regard their analogue of full-rank correlation as a reasonable requirement and center their discussion around the corresponding impossibility result, we argue above that full-rank correlation is not necessarily economically natural and devote the remainder of the paper to the surprising richness of communication that is possible if there is at least some independence in the agents’ signals.\footnote{KST also consider a version of our binary-decision model but there maintain their analogue of full-rank correlation. Hence, they obtain none of our subsequent positive results for that model.}

\section{The Role of Interdependence}

While correlation between the agents' signals restricts the principal's ability to elicit information from the agents, we now show that interdependence of the agents' signals with the principal's payoffs allows the principal to benefit from communication.
Recall that a BIC mechanism $x$ is said to be \emph{profitable} if $\E_{s}[x(\bm{s}) v(\bm{s})] > \max\lbrace 0, \E_{s}[v(\bm{s})]\rbrace$, i.e., under $x$ the principal does strictly better than without information.

Payoff-interdependence here means that the principal's expected payoff is \emph{not} additively separable in the agents' signals (non-additive for short). For the definition, recall that $v(s_{\ell}, s_{r}) = \E_{\omega}[v(\bm{\omega})\mid s_{\ell}, s_{r}]$ is the principal's expected payoff from decision $L$ given signals $(s_{\ell}, s_{r})$, while the payoff from decision $R$ is normalized to zero.

\begin{definition}\label{def:additive}
The principal's payoff $v$ is \emph{additive} if there exist integrable $v_{\ell}\colon S_{\ell}\to \R$ and $v_{r}\colon S_{r}\to\R$ such that, for $\pi$-almost all $s_{\ell}$ and $s_{r}$,
\begin{equation*}
    v(s_{\ell}, s_{r}) = v_{\ell}(s_{\ell}) + v_{r}(s_{r}).
\end{equation*}
\end{definition}
To interpret additivity, consider the impact on the principal's payoff of changing agent $\ell$'s signal, $v(s_{\ell}, s_{r}) - v(s_{\ell}', s_{r})$.
If $v$ is additive, then this impact is independent of agent $r$'s signal $s_{r}$.
Conversely, if $v$ is non-additive, then at least one agent's signal is cross-relevant for the impact of the other agent's signal on the principal's payoff.
The introductory example in \Cref{sec:intro_example} has the non-additive payoff $v(s_{\ell}, s_{r}) = \mathbbm{1}(s_{\ell} > s_{r}) - \mathbbm{1}(s_{\ell} \leq s_{r})$ with $s_{\ell}, s_{r} \in [0,1]$.
Many other environments generate non-additive payoffs; see \Cref{sec:optimal_mechanism,sec:no_commitment} for examples. 

We begin by observing that non-additivity is \emph{necessary} for the principal to profitably elicit information.
Recall that in a BIC mechanism $x$ the interim expected decision probabilities must be constant, i.e., $\E_{s_{-i}}[x(s_{i}, \bm{s_{-i}}) \mid s_{i}] = \bar{x}$ for all $i$ and $\pi_{i}$-almost all $s_{i}$ (\Cref{lemma:BIC_characterization}).
Now suppose the principal's payoff is additive, $v = v_{\ell} + v_{r}$.
We can then write the principal's utility from $x$ as:
\begin{align*}
    \E_{s}[x(\bm{s}) v(\bm{s})] 
    &= \E_{s}[x(\bm{s}) (v_{\ell}(\bm{s_{\ell}}) + v_{r}(\bm{s_{r}}))] \\
    &= \E_{s_{\ell}}[\E_{s_{r}}[x(\bm{s})\mid \bm{s}_{\ell}]v_{\ell}(\bm{s_{\ell}}) ] +  \E_{s_{r}}[ \E_{s_{\ell}}[x(\bm{s})\mid \bm{s}_{r}]v_{r}(\bm{s_{r}})] \\
    &= \bar{x}\left(\E_{s_{\ell}}[v_{\ell}(\bm{s_{\ell}})] + \E_{s_{r}}[v_{r}(\bm{s_{r}})]\right)  \\
    &= \bar{x}\E_{s}[v(\bm{s})] 
    .
\end{align*}
This is also the expected utility from not consulting the agents and simply taking decision $L$ with probability $\bar{x}$.
Hence, there is no strict gain from communicating with the agents, i.e., non-additivity is necessary for profitable communication.

For the converse, we focus here on independent signals as this case admits the largest set of BIC mechanisms (\Cref{proposition:spanning_order}) and thus helps isolate the role of payoff-interdependence. Nevertheless,  our results can be extended to the intermediate case where signals are correlated but do not exhibit full-rank correlation: it suffices for non-additivity to hold in dimensions (or directions) along which the agents' signals are independent. See the earlier version of this paper for the details (\citet{kattwinkel2022mechanisms}). 

Our main insight in this section is that, when the principal is \emph{unbiased}, i.e., ex-ante indifferent between $L$ and $R$ ($\E_{s}[v(\bm{s})] = 0$), then non-additivity is also sufficient for profitable communication.

\begin{proposition}\label{proposition:profitability_independence}
    Suppose the agents' signals are independent and the principal is unbiased. Then, a profitable BIC mechanism exists if and only if $v$ is non-additive.
\end{proposition}
\begin{proof}
Necessity follows from the argument above. For sufficiency, see \Cref{appendix:proof_profitability_independence}.
\end{proof}
The proof constructs profitable BIC mechanisms in closed form.
In \Cref{sec:optimal_mechanism}, we further generalize this insight to the case when the principal is biased. We also give examples of profitable mechanisms and further illustrate the role of payoff additivity.

\section{Optimal Mechanisms}\label{sec:optimal_mechanism}
We now show that optimal BIC mechanisms can be characterized in the Myersonian tradition via an ironed virtual value.
In the setting of \citet{myerson1981optimal}, the principal's utility from a mechanism can be expressed in terms of the principal's decision (an allocation) and virtual values. BIC in that setting imposes monotonicity constraints on the decision. Myerson develops an ironing procedure ensuring that, if one implements the first-best decision under the ironed virtual value, then one obtains a decision rule that satisfies the monotonicity constraints and, moreover, defines an optimal mechanism. We analogously define an ironed virtual value $v^{\ast}$ for our setting ensuring that, if one takes the first-best decision under $v^{\ast}$, then one obtains an optimal BIC mechanism. 
We again consider the case with independent signals for simplicity. 

To define virtual values, we consider alternative payoffs $v'\colon S_{\ell}\times S_{r}\to\R$ to the principal's true payoff $v\colon S_{\ell}\times S_{r}\to\R$.
A payoff $v'$ is \emph{BIC-equivalent} to $v$ if, for all BIC mechanisms $x$, both payoffs induce the same expected utility under $x$, i.e., $\E_{s}[v(\bm{s}) x(\bm{s})] = \E_{s}[v'(\bm{s}) x(\bm{s})]$.
A payoff $v'$ is BIC-equivalent to $v$ if and only if the difference $v - v'$ has mean zero and is additive;
to see this, apply \Cref{proposition:profitability_independence} with the fictitious payoff $v - v'$.  Let $\mathcal{V}$ be the set of BIC-equivalent payoffs.

We define a BIC-equivalent payoff $v^{\ast}\in\mathcal{V}$ as a virtual value if it minimizes the principal's utility in a fictitious scenario where the principal can choose optimally ex-post, ignoring all BIC constraints. The principal's fictitious decisions then turn out to define an optimal BIC mechanism under an additional regularity condition.
Formally, $v^{\ast}\in\mathcal{V}$ is a \emph{virtual value} if
\begin{equation*}
    v^{\ast} \in \argmin_{v'\in \mathcal{V}} \E_{s}[\max\lbrace 0, v'(\bm{s})\rbrace].
\end{equation*}

We impose the following regularity condition.
\begin{assumption}[Regularity]\label{assumption:regularity}
    The signal spaces $S_{\ell}$ and $S_{r}$ are compact,
    and $v$ is \emph{nowhere additive}, i.e., for all additive functions $w$ it holds $v = w$ only on a set of $\pi$-measure zero.
\end{assumption}
The compactness assumption ensures that a virtual value exists, while the assumption of nowhere additivity ensures that the point-wise virtual value maximizer defines the unique optimal BIC mechanism.\footnote{As can be seen from the proof of \Cref{theorem:optimal_mech}, if nowhere additivity does not hold, then any point-wise virtual value maximizer still attains the same expected virtual value as the optimal BIC mechanism, but not every point-wise maximizer need be BIC. The additional complication is that virtual values may be zero on a set of strictly positive measure. A judicious tie-breaking rule on that set would be needed to ensure BIC. Our arguments show that such a tie-breaking rule exists.}
For instance, if $s_{\ell}$ and $s_{r}$ are both in $[0,1]^m$, $\pi$ has a density (with respect to Lebesgue measure), and $v$ is smooth, then nowhere additivity holds if, for almost every signal profile $(s_\ell,s_r)$, the cross-partial derivative of $v$ with respect to some component of $s_{\ell}$ and some component of $s_{r}$ is non-zero. In plain words, nowhere additivity only requires that there is at least some signal interdependence almost everywhere; see \Cref{appendix:sufficient_for_nonadditive} for further details.

\begin{theorem}\label{theorem:optimal_mech}
    Let signals be independent. Let \hyperref[assumption:regularity]{Regularity} hold.
    There exists a virtual value $v^{\ast}$. 
    The optimal BIC mechanism $x^{\ast}$ is $\pi$-almost everywhere unique and maximizes the virtual value point-wise, $x^{\ast}(s) = \mathbbm{1}(v^{\ast}(s) \geq 0)$ for $\pi$-almost all $s$.
    The principal's optimal utility is $\E_{s}[\max\lbrace 0, v^{\ast}(\bm{s})\rbrace]$.
\end{theorem}
\begin{proof}
See \Cref{appendix:proof_optimal_mech}.
\end{proof}

Since virtual values are BIC-equivalent to $v$, they can be written as $v^{\ast}(s) = v(s) - v^{\ast}_{\ell}(s_{\ell}) - v^{\ast}_{r} (s_{r})$ for all $s$, where $v^{\ast}_{\ell} + v^{\ast}_{r}$ has mean zero. Thus, the principal can only enjoy that part $v^{\ast}$ of the payoff $v$ that is not explained by additive terms. Intuitively, the principal can exploit exactly the interdependencies between the agents' signals in their payoff.

\Cref{theorem:optimal_mech} also yields, under \hyperref[assumption:regularity]{Regularity}, a generalization of \Cref{proposition:profitability_independence} to situations in which the principal is not unbiased.
The principal's profit from using an optimal BIC mechanism rather than taking the ex-ante optimal decision equals:\footnote{To see this, recall that all $a\in\R$ satisfy $2\max\lbrace 0, a\rbrace = \vert a\vert + a$, and that $\E_{s}[v^{\ast}(\bm{s})] = \E_{s}[v(\bm{s})]$ holds by definition of $\mathcal{V}$. So, $2\E_{s}[\max\lbrace 0, v^{\ast}(\bm{s})\rbrace] - 2\max\lbrace 0, \E_{s}[v(\bm{s})]\rbrace = \E_{s}[\vert v^{\ast}(\bm{s})\vert] - \vert \E_{s}[v(\bm{s})]\vert + \E_{s}[v^{\ast}(\bm{s})] - \E_{s}[v(\bm{s})] = \E_{s}[\vert v^{\ast}(\bm{s})\vert] - \vert \E_{s}[v(\bm{s})]\vert$.}
\begin{equation*}
\frac{\E_{s}[\vert v^{\ast}(\bm{s})\vert] - \vert \E_{s}[v(\bm{s})]\vert}{2}
.
\end{equation*}
That is, communication is profitable if and only if the non-additive part $\E_{s}[\vert v^{\ast}(\bm{s})\vert]$ is larger than the principal's ex-ante bias $\vert \E_{s}[v(\bm{s})]\vert$ for one of the decisions.

To prove \Cref{theorem:optimal_mech}, we build on techniques from the literature on \emph{capacity-constrained optimal transport} (\citet{korman2015optimal}; \citet{korman2015dual}). Consider BIC mechanisms $x$ for which the ex-ante decision probability $\E_{s}[x(\bm{s})]$ equals some fixed probability $\bar{x}$. 
Together with the signal distribution $\pi$, the mechanism $x$ induces an endogenous density $x \de \pi$ on $S_{\ell}\times S_{r}$. Under independence, our BIC characterization says that the marginals of this density must be $\bar{x} \de \pi_{\ell}$ and $\bar{x} \de\pi_{r}$ (\Cref{lemma:BIC_characterization}). The difference to classical optimal transport is that the density $x\de\pi$ is bounded above since the mechanism specifies a decision probability, i.e., $0\leq x \leq 1$. 
Thus, we obtain a capacity-constrained optimal transport problem. \citet{korman2015dual} prove dual attainment for such transport problems. We adapt their techniques to our specific problem in which $\bar{x}$ is also a choice variable.\citet{korman2015optimal} study when a capacity-constrained transport problem admits a unique solution. In our problem, where $\bar x$ is endogenous, we instead derive uniqueness from nowhere additivity and the structure of the problem for the optimal $\bar{x}$.

The optimal BIC mechanism $x^{\ast}$ is deterministic, which leverages that the signal distribution is atomless (as an implicit consequence of \hyperref[assumption:regularity]{Regularity}).\footnote{The restriction of $v$ to a single point is additive, so no single point may admit mass. If we only wanted to show the existence of a deterministic optimal BIC mechanism when the signal distribution has atomless marginals, there is a shorter proof using the purification result of \citet[Theorem 2]{chen2019equivalence}. But this purification result tells us nothing about the form of the optimal mechanism.} If signals are discrete, randomization is typically necessary to satisfy BIC: for generic marginal signal distributions, the only deterministic BIC mechanisms are constantly zero or constantly one. As signals become finer, the fraction of signal profiles at which one must randomize to satisfy BIC vanishes, which is consistent with \Cref{theorem:optimal_mech}. 
See \Cref{appendix:discrete_randomization} for a characterization of the set of optimal mechanisms in the discrete case.

\subsection*{Further Examples}

In some cases, we can use a simple guess-and-verify approach to find a virtual value and the unique optimal BIC mechanism.
Suppose again that signals are independent.

Consider the modified payoff $\tilde{v}$, defined for all $s_{\ell}$ and $s_{r}$ by
\begin{equation*}
    \tilde{v}(s_{\ell}, s_{r}) = v(s_{\ell}, s_{r}) - \E_{s_{r}}[v(s_{\ell}, \bm{s}_{r})] - \E_{s_{\ell}}[v(\bm{s}_{\ell}, s_{r})] + 2 \mathbb{E}_{s}[v(\bm{s})].
\end{equation*}
Notice that $\tilde{v} - v$ is additive and has mean-zero, meaning that the payoff $\tilde{v}$ is BIC-equivalent to the true payoff $v$.
Therefore, for all BIC mechanisms, the principal's utility from $x$ is bounded above as follows:
\begin{align*}
    \mathbb{E}_{s}[v(\bm{s}) x(\bm{s})] = \mathbb{E}_{s}[\tilde{v}(\bm{s}) x(\bm{s})] \leq \mathbb{E}_{s}[\max\left\lbrace 0, \tilde{v}(\bm{s})\right\rbrace]
    .
\end{align*}
The right-hand side is simply the principal's virtual utility from the mechanism $\tilde{x}$ that is given by $\tilde{x} (s) = \mathbbm{1}(\tilde{v}(s) \geq 0)$ for all $s$.
It follows that if $\tilde{x}$ is BIC, then $\tilde{x}$ is an optimal BIC mechanism, $\tilde{v}$ is a virtual value, and the right-hand side is the principal's utility from $\tilde x$ (because $\tilde v$ is BIC-equivalent to $v$).\footnote{These conclusions remain true for every BIC mechanism that agrees with $\tilde{x}$ wherever $\tilde{v}\neq 0$.}
The mechanism $\tilde x$ is BIC if and only if the probabilities $\mathbb{P}_{s_{r}}\left( \tilde{v}(s_{\ell}, \bm{s_{r}}) \geq 0\right)$ and $\mathbb{P}_{s_{\ell}}\left( \tilde{v}(\bm{s_{\ell}}, s_{r}) \geq 0\right)$, respectively, are constant in $s_{\ell}$ and $s_{r}$, respectively. If, moreover, $\tilde{x}$ is BIC and $\tilde{v}\neq 0$ holds $\pi$-almost everywhere, then $\tilde{x}$ is the $\pi$-almost everywhere unique optimal BIC mechanism.

This observation can be used to quickly obtain an optimal BIC mechanism in several (but not all) examples. This approach also does not require \hyperref[assumption:regularity]{Regularity}.

\begin{example}\label{example:revisited}
    \Cref{sec:intro_example} introduces an example with $s_{\ell}, s_{r} \in [0, 1]$ and $v(s_{\ell}, s_{r}) = \sgn (s_{\ell} - s_{r})$. 
    Let $s_{\ell}$ and $s_{r}$ be independently and uniformly distributed.
    In this case, $\tilde{v}(s_{\ell}, s_{r}) = \sgn(s_{\ell} - s_{r}) - 2(s_{\ell} - s_{r})$, which is weakly positive if and only if $s_{\ell} - s_{r} \leq -1/2$ or $0\leq s_{\ell} - s_{r} \leq 1/2$. One may now check that both $\mathbb{P}_{s_{r}}\left( \tilde{v}(s_{\ell}, \bm{s_{r}}) \geq 0\right)$ and $\mathbb{P}_{s_{\ell}}\left( \tilde{v}(\bm{s_{\ell}}, s_{r}) \geq 0\right)$ are constantly equal to $1/2$, as visualized in \Cref{figure:horse_race} in \Cref{sec:intro_example}. Thus, $\tilde{x}$ as defined above is the unique optimal BIC mechanism. As described in \Cref{sec:intro_example}, this mechanism sides with the party $i$ that makes the stronger case $s_{i}$ unless the gap between both cases exceeds one half, in which case the mechanism sides with the party that makes the weaker case.

    If the signal distribution is given by some strictly increasing, atomless c.d.f.\ $F$, then an optimal BIC mechanism is still given by $\tilde{x} (s) = \mathbbm{1}(\tilde{v}(s) \geq 0)$ for all $s$, where now $\tilde{v}(s) = \sgn(s_{\ell} - s_{r}) - 2 (F(s_{\ell}) - F(s_{r}))$.
    To see this, let $x$ be an arbitrary BIC mechanism. By relabeling signals, we obtain a new mechanism $y$ defined on quantiles: $y(q_{\ell}, q_{r}) = x(F^{-1}(q_{\ell}), F^{-1}(q_{r}))$ for all quantiles $q_{\ell}, q_{r}\in [0, 1]$. Clearly, $y$ is BIC in a setting with uniformly distributed $q_{\ell}$ and $q_{r}$.
    Moreover, the relabeling does not affect the principal's preferences over mechanisms since $\sgn(s_{\ell} - s_{r}) = \sgn(q_{\ell} - q_{r})$ whenever $s_{\ell} = F^{-1}(q_{\ell})$ and $s_{r} = F^{-1}(q_{r})$.
    In the setting with uniformly distributed $q_{\ell}$ and $q_{r}$, the above argument shows that an optimal BIC mechanism is given by $\mathbbm{1}(\sgn(q_{\ell} - q_{r}) - 2(q_{\ell} - q_{r}) \geq 0)$ for all $q_{\ell}, q_{r}$.
    Hence, $\tilde{x}(s) = \mathbbm{1}(\sgn(s_{\ell} - s_{r}) - 2 (F(s_{\ell}) - F(s_{r})) \geq 0)$ defines an optimal BIC mechanism in the original setting.
\end{example}

\begin{example}\label{example:misfit}
    A decision-maker decides whether to integrate two separate entities. One entity prefers integration while the other opposes integration. Each entity has private information $s_i\in [0, 1]$ on issues that matter for the fit between the two. The decision-maker's payoff from integration equals some cost savings $c$ minus a misfit cost:
    \begin{equation*}
        v(s_{\ell}, s_{r}) = c - \frac{(s_{\ell} - s_{r})^{2}}{2}.
    \end{equation*}

    For example, a government decides whether to amalgamate two municipalities, where the combined administration would be seated in one municipality or where the poorer municipality would be cross-subsidized by the other; a board decides whether to merge two divisions, where one division's manager would run the combined division; a firm decides whether to integrate external consultants into its in-house consulting unit, so that the in-house unit would lose influence within the firm;
    or a regulator decides whether to impose a common standard on two firms, where one firm earns rents from a proprietary standard and the other does not.

    Let $s_{\ell}$ and $s_{r}$ be independent and have atomless distributions with respective means $\mu_{\ell}$ and $\mu_{r}$ and variances $\sigma_{\ell}^{2}$ and $\sigma_{r}^{2}$.
    Suppose $c$ is such that the decision-maker is ex-ante indifferent, meaning $2c=\sigma_\ell^2+\sigma_r^2+(\mu_\ell-\mu_r)^2$.
    Then, one may calculate: 
    \begin{equation*}
        \tilde{v}(s_{\ell}, s_{r}) = \left(s_{\ell} - \mu_{\ell}\right) \left(s_{r} - \mu_{r}\right).
    \end{equation*}
    The candidate mechanism is thus:
    \begin{quote}
    \emph{each entity $i=\ell,r$ separately reports a position $\hat s_i\in[0,1]$ to the decision-maker; the decision-maker integrates if both positions deviate in the same direction from the expected positions, i.e., $(\hat s_\ell-\mu_\ell)(\hat s_r-\mu_r)> 0$, keeps the entities separate if they deviate in opposite directions, and integrates with probability $\tfrac12$ if a reported position $\hat s_i$ equals the expected position $\mu_i$.}
    \end{quote}
    This mechanism is BIC and hence uniquely optimal if and only if, for both $i$, the mean $\mu_{i}$ equals the median of the signal distribution.\footnote{To see this, suppose entity $-i$ reports truthfully. Reporting a position $\hat s_i$ above the mean yields an expected probability of integration of $\mathbb{P}_{s_{-i}}(\bm{s_{-i}}>\mu_{-i})$, reporting a position below the mean yields $\mathbb{P}_{s_{-i}}(\bm{s_{-i}}<\mu_{-i})$, and reporting the mean itself yields $\tfrac12$. These are all equal if and only if $\mu_{-i}$ is the median of $s_{-i}$.}
    For instance, let $s_{\ell}$ and $s_{r}$ be symmetric around their means.

\end{example}

As \Cref{example:misfit} shows, the guess-and-verify approach does not always work. In \Cref{appendix:CES_production}, we solve an additional example where we identify a virtual value that is different from $\tilde v$ and explain how it defines an optimal BIC mechanism.

\section{Communication without Commitment}\label{sec:no_commitment}
The optimal BIC mechanisms analyzed so far represent an upper bound on the information that the principal can elicit and use.
In some cases, using an optimal BIC mechanism requires the principal to commit to taking the ex-post incorrect decision at certain posterior beliefs; see, e.g., the optimal BIC mechanism in \Cref{sec:intro_example}. As discussed in the Introduction, such commitment can be reasonably achieved in many circumstances through reputation or law. However, one might wonder whether communication crucially relies on commitment.

We here show that the principal can benefit from communication even without commitment to a mechanism.
That is, we now consider communication games where the agents simultaneously submit messages to the principal, and the principal then takes an ex-post optimal decision based on their posterior belief. We refer to a perfect Bayesian equilibrium (PBE) of this communication game as a \textit{mechanism without commitment}. We also discuss how the principal's lack of commitment can be overcome through a mediator.

Our sufficient condition for beneficial communication is not knife-edge. Specifically, the condition is that the principal's payoff be \emph{non}-separable, where separability is defined as follows.
\begin{definition}\label{def:separable}
    The principal's expected payoff $v$ is \textit{separable} if there exist integrable functions $v_{\ell}^{1}, v_{\ell}^{2} \colon S_\ell\to\R$ and $v_{r}^{1}, v_{r}^{2}\colon S_{r}\to\R$ such that, for $\pi$-almost all $s_\ell\in S_\ell$ and $s_r\in S_r$,
    \begin{equation*}
        v(s_\ell,s_r)= v_{\ell}^{1}(s_{\ell}) v_{r}^{1}(s_{r}) + v_{\ell}^{2}(s_{\ell}) v_{r}^{2}(s_{r})
        .
    \end{equation*}
\end{definition}
Separable payoffs include additively separable payoffs (\Cref{def:additive}) and multiplicatively separable payoffs. Non-separable payoffs include, for example:
\begin{itemize}
    \item Picking the agent with the better claim, $v(s) = \sgn(s_{\ell} - s_{r})$, as in \Cref{sec:intro_example}.
    \item Fit quality, $v(s) = c - (s_{\ell} - s_{r})^{2}/2$, as in \Cref{example:misfit}.
    \item CES-production, $v(s) =((s_\ell^\rho+s_r^\rho)/2)^{1/\rho}$ for $\rho \notin \lbrace 0, 1\rbrace$. For $\rho \to 0$, we would obtain Cobb-Douglas production, which is multiplicatively separable. For $\rho = 1$, the payoff is additively separable.  
\end{itemize}

\begin{theorem}\label{theorem:no_commitment_profitable}
    If the agents' signals are independent, the principal is unbiased, and $v$ is non-separable, then there exists a profitable mechanism without commitment.
\end{theorem}
\begin{proof}
See \Cref{appendix:proof_no_commitment_profitable}.
\end{proof}

For the PBEs constructed here, we consider games where each agent sends one of two possible messages, inducing a binary garbling of each agent's signal.
When the principal's payoff is non-separable, there are binary garblings of each agent's signal that are informative to the principal only jointly, but not on their own. 
We use the fact that individual messages are uninformative to verify the agents' and the principal's incentives in the candidate PBE. The fact that messages are jointly informative implies that the mechanism is overall profitable for the principal. 
The same construction works if the principal has a small bias towards one decision.\footnote{Precisely, fix a non-separable, unbiased $v$. The mechanism without commitment constructed in the proof of \Cref{theorem:no_commitment_profitable} gives the principal a constant expected payoff of $V>0$ conditional on taking decision $L$. Now suppose the principal's payoff is instead $v + b$, where $b\in\R$ is a bias. If $V > \vert b\vert$, then the mechanism without commitment constructed in the proof remains a PBE and profitable.}

\begin{remark}
    Profitable communication may also be possible without commitment when $v$ is separable, i.e., non-separability is not necessary. Suppose $s_{\ell}$ and $s_{r}$ are independent and uniformly distributed in $[-1, 1]$, and $v(s)$ equals $1$ if $s_{\ell}$ and $s_{r}$ have the same sign but equals $-1$ if they have opposite signs. This payoff is separable.\footnote{Indeed, $v(s) = \mathbbm{1}(s_{\ell} \geq 0)(\mathbbm{1}(s_{r} \geq 0) - \mathbbm{1}(s_{r} < 0)) + \mathbbm{1}(s_{\ell} < 0)(\mathbbm{1}(s_{r} < 0) - \mathbbm{1}(s_{r} \geq 0))$ a.e.} Nevertheless, it is easy to see that the principal obtains the ex-post optimal outcome in a PBE in which each agent truthfully reports whether their signal is positive.
\end{remark}

\begin{remark}
  If all parties can communicate through an intermediary, then a mechanism is implementable despite the principal's lack of commitment if and only if it is BIC and gives the principal at least their no-information payoff $\max\lbrace 0, \E[v(\bm{s})]\rbrace$. In particular, every optimal or profitable BIC mechanism is implementable. To see this, note that by the communication revelation principle (\citet{myerson1982optimal}), it suffices to consider equilibria in which the agents truthfully report their signals $s = (s_{\ell}, s_{r})$ to the intermediary, the intermediary recommends decision $L$ with probability $x(s)$ and $R$ otherwise, and the principal follows the recommendation. The principal's obedience constraints require that, after recommendation $L$, taking $L$ is optimal, $\E_{s}[v(\bm{s})x(\bm{s})] \geq 0$, and that, after recommendation $R$, taking $R$ is optimal, $0 \geq \E_{s}[v(\bm{s})(1-x(\bm{s}))]$. This rearranges to $\E_{s}[v(\bm{s})x(\bm{s})] \geq \max\lbrace 0, \E_{s}[v(\bm{s})]\rbrace$. Thus, the principal's obedience constraints hold if and only if $x$ gives the principal at least their no-information payoff. If the principal follows the intermediary's recommendation, then the agents' payoffs are exactly as in the mechanism $x$, so truthful reporting is optimal if and only if $x$ is BIC.
\end{remark}

\section{Discussion}\label{sec:discussion}
We set out to understand whether information can be elicited from opposing factions even in those adversarial situations in which the principal cannot rely on partial alignment, compromise, or any additional instruments. 
We have shown that there is indeed substantial scope for communication in these situations and have characterized when and how information can be elicited.

We conclude by discussing environments that are less adversarial to the principal. We 
consider partially aligned agents, multiple alternatives, factions that consist of multiple agents, as well as several competing factions.
We skip any discussion of additional screening devices such as transfers or costly audits because they are obviously beneficial to the principal and have been amply studied in the literature (see, e.g., \citet{Cremer.1988,ben2014optimal}, respectively).

\subsection{Partial Alignment}
In the baseline model, the agents' preferences are independent of the state and, in that sense, unaligned with the principal's preferences.
Whereas additional agents or alternatives will be unambiguously beneficial for the principal, the impact of partial alignment is more nuanced. We show that weak alignment can make information elicitation more difficult relative to the baseline setting. Meanwhile, if alignment is sufficiently strong, then it is always possible to elicit useful information even when it would be impossible in the baseline model.

Formally, suppose agent $\ell$'s payoff from decision $L$ is now given by $u_{\ell}(s_{\ell}, s_{r}) = 1 + \alpha v(s_{\ell}, s_{r})$ and agent $r$'s payoff from $L$ is $u_{r}(s_{\ell}, s_{r}) = -1 + \alpha v(s_{\ell}, s_{r})$. The payoffs from $R$ remain normalized to zero. Here, $\alpha \geq 0$ captures the alignment between the principal's and the agents' preferences. The baseline model corresponds to the case without alignment, $\alpha = 0$.

We first show that, perhaps surprisingly, designing a profitable mechanism is easier in the baseline model without alignment than under weak alignment. In particular, the principal can be better off when preferences are not aligned. 
We say that alignment is \emph{weak} if $\E_{s_{r}}[u_{\ell}(s_{\ell}, \bm{s_{r}})\mid s_{\ell}] > 0 > \E_{s_{\ell}}[u_{r}(\bm{s_{\ell}}, s_{r})\mid s_{r}]$ holds for $\pi_{\ell}$-almost all $s_{\ell}$ and $\pi_{r}$-almost all $s_{r}$, i.e., each agent is still interim biased towards one decision.
\begin{proposition}\label{proposition:partial_alignment}
    Let the agents' signals be independent. Let the principal be unbiased.
    If there exists a profitable BIC mechanism under weak alignment, then there exists a profitable BIC mechanism without alignment; the converse is not generally true.
\end{proposition}
\begin{proof}
See \Cref{appendix:proof_partial_alignment}.
\end{proof}

For intuition, recall that BIC mechanisms may provide incentives by committing the principal to decisions that are wrong ex post. With partial alignment, however, the agents bear part of the principal's loss from such decisions. An agent may therefore misreport to prevent a wrong decision even if this misreport does not change the agent's interim-expected decision. Thus, alignment can undermine the incentives provided by commitment to wrong decisions. Notice the contrast to standard communication models \`a la \citet{crawford1982strategic} and delegation models \`a la \citet{holmstrom1977incentives}, in which the principal is better off when the agent is more aligned.

Sufficiently strong alignment, by contrast, facilitates information elicitation. Suppose the interim preference of at least one agent flips depending on their signal, so that alignment is not weak. Say (without loss) this is true for agent $\ell$, i.e., $\E_{s_{r}}[u_{\ell}(s_{\ell}, \bm{s_{r}})\mid s_{\ell}] < 0$ with strictly positive probability. Then, full delegation to agent $\ell$, i.e., $x(s_{\ell}) = \mathbbm{1}(1 + \alpha \E_{s_{r}}[v(s_{\ell}, \bm{s_{r}})\mid s_{\ell}]\geq 0)$, is a profitable BIC mechanism when the principal is unbiased. Relative to always choosing $L$, which yields zero to an unbiased principal, this mechanism switches to $R$ exactly when agent $\ell$ prefers $R$ given their signal, and, in this case, the principal prefers $R$ as well. Notice that this mechanism is profitable even if the principal’s payoff is additive. Moreover, the same outcome obtains through cheap talk: agent $\ell$ recommends their preferred decision to the principal, and it is sequentially rational for the principal to follow the recommendation.

This discussion implies that the principal's utility from an optimal BIC mechanism may be non-monotonic as a function of alignment $\alpha$.

\subsection{Multiple Alternatives}\label{sec:discussion:multiple_alternatives} 
When the principal's decision space is richer, there is more scope for communication. To illustrate this, we show that profitable mechanisms now exist even under full-rank correlation. To highlight the contrast with the baseline model and to separate the effect of multiple alternatives from that of partial alignment, we maintain the zero-sum structure on the agents' preferences. 

Suppose the space of decisions is a finite set $A$. Agent $\ell$ has a Bernoulli utility function $u\colon A\to\mathbb{R}$, while agent $r$ has Bernoulli utility function $-u$. The principal has a state-dependent Bernoulli utility function $v\colon A\times S_{\ell}\times S_{r}\to\mathbb{R}$. Without loss of generality, the utility from one decision can be normalized to zero, i.e., there exists $a_{0}\in A$ such that $v(a_{0}, s_{\ell}, s_{r}) = u(a_{0}) = 0$ for all $(s_{\ell}, s_{r})$. A mechanism is a function $x\colon S_{\ell}\times S_{r}\to \Delta A$. Agent $i$'s expected utility from reporting $s_{i}'$ given true signal $s_{i}$ equals $\pm\mathbb{E}_{s_{-i}}[\sum_{a} u(a) x(a, s_{i}', \bm{s_{-i}})\mid s_{i}]$, with $+$ for $\ell$ and $-$ for $r$. The principal's expected utility is $\mathbb{E}_{s}[\sum_{a} v(a, \bm{s}) x(a, \bm{s})]$.

The BIC characterization is analogous to the binary-decision setting. 
There, we identified agents' payoffs with the probability $x$ of decision $L$: agent $\ell$ obtains $x$, while agent $r$ obtains $-x$. Our binary-decision BIC characterization says that, for each agent, the expected payoff must be constant in their true signal and their reported signal. This is also exactly the BIC characterization for multiple decisions. That is, if a mechanism $x$ is BIC, then, for all $i$, the expectation $\mathbb{E}_{s_{-i}}[\sum_{a} u(a) x(a, s_{i}', \bm{s_{-i}})\mid s_{i}]$ is the same for $\pi_{i}$-almost all $s_{i}$ and $\pi_{i}$-almost all $s_{i}'$; conversely, if this holds, then $x$ coincides $\pi$-almost everywhere with a BIC mechanism. The proof is analogous to the binary-decision case.\footnote{Indeed, with $y(s) = \sum_{a} u(a) x(a, s)$, which takes values in $[m, M]$ for $m = \min_{a} u(a)$ and $M = \max_{a} u(a)$, $(y - m)/(M - m)$ is a binary-decision mechanism that is BIC if and only if $x$ is BIC (if $m=M$, every mechanism is trivially BIC).}  

The BIC characterization implies that, as in the binary-decision setting, the set of BIC mechanisms shrinks in the spanning order.

We now show that having access to a richer decision space $A$ lets the principal more easily elicit information compared to the binary-decision case.
To do so, we suppose for simplicity that $S_{\ell}$ and $S_{r}$ are finite. To stack the cards against the principal, let us also assume that the signal distribution $\pi$ has full rank.
The BIC characterization in this case is that $\sum_{a\in A} u(a) x(a, s)$ is constant in $s=(s_\ell,s_r)$. 
With binary decisions, this would imply that the overall mechanism is constant (\Cref{proposition:spanning_order}). With multiple alternatives, this need not be the case. The following result gives a necessary and sufficient condition for the principal to profitably elicit information when the principal is \textit{unbiased}, i.e., ex-ante indifferent between all alternatives ($\mathbb{E}_{s}[v(a, \bm{s})] = 0$ for all $a\in A$).
As before, we say that a BIC mechanism $x$ is \emph{profitable} if $\mathbb{E}_{s}[\sum_{a} v(a, \bm{s}) x(a, \bm{s})] > \max_{a\in A} \mathbb{E}_{s}[v(a, \bm{s})]$.
\begin{proposition}\label{proposition:multiple_alternatives_profitability}
    Suppose the principal is unbiased and the signal distribution $\pi$ has full support and full rank. Then, a profitable mechanism exists if and only if there does \emph{not} exist a function $w\colon S_{\ell}\times S_{r}\to\mathbb{R}$ such that, for all $(s_\ell,s_r)\in S_\ell\times S_r$ and all $a \neq a_{0}$,
    \begin{equation*}
        v(a, s) = u(a) w(s_\ell,s_r).
    \end{equation*}
\end{proposition}
\begin{proof}
See \Cref{appendix:proof_multiple_alternatives_profitability}.
\end{proof}
Such a function $w$ exists if and only if, at every signal profile, either the principal's vNM preferences over lotteries in $\Delta A$ fully coincide with the vNM preferences of one of the agents or the principal is indifferent among all alternatives. Thus, the condition does not hold---and hence a profitable mechanism exists---if, at \textit{some} signal profile, the principal is not indifferent among all alternatives and the principal's vNM preferences coincide with neither agent's preferences.

This result echoes the intuition from the cheap-talk literature that indifferences help elicit information from an expert with biased, state-independent preferences (\cite{chakraborty2010persuasion,lipnowski2020cheap}).

\subsection{Two Factions with Internal Coordination}\label{sec:discussion:teams}
In some applications, the two opposing factions may each be composed of multiple individuals. For example, in a labor dispute, the principal arbitrates between a labor union and a firm. If the individuals within each faction all advocate for the same decision, our analysis applies. 
First, suppose the principal cannot control communication within each faction. For example, the arbitrator cannot control the flow of emails within the union or the firm. Then, the members of the same faction plausibly collude when reporting their private information to the principal's mechanism to maximize the probability of their faction's preferred decision. Thus, it is as if each faction was represented by a single agent. Consequently, we obtain our baseline model.
If the factions cannot coordinate internally, then it is even easier for the principal to elicit some information: in the setting of \Cref{proposition:profitability_independence}, it suffices that the principal's payoff is non-additive in the signals of two agents (from the same or different factions).\footnote{If the principal's mechanism generates constant interim-expectations, as in our BIC characterization, it does not matter whether the two agents come from different factions, i.e., whether or not they have different preferences.}

\subsection{Multiple Competing Agents} The presence of additional agents increases the scope for communication.
For concreteness, suppose the principal allocates a good such as a budget to one of multiple agents. Each agent's payoff is simply their probability of getting the good. 

If the agents' signals are correlated and there are at least three agents, then the principal can use this correlation to cross-check the agents' reports. 
For example, the principal can designate one agent as a juror who never gets the good but whose information helps decide which of the remaining agents should get the good. Since signals are correlated, the juror's signal can further be used to cross-check the reports of the remaining agents, which relaxes the remaining agents' incentive constraints.
This cross-checking complicates the BIC characterization and suggests Cr{\'e}mer--McLean-like constructions. Of course, since the principal cannot use monetary transfers, such constructions would involve the allocation itself, in contrast to Cr{\'e}mer and McLean's side bets.
Analogous constructions cease to be possible in our baseline setting (with two agents and where the principal must allocate) by \Cref{proposition:spanning_order}.

We do not expect an economically meaningful characterization of optimal BIC mechanisms, even if signals are independent. This is because, in the allocation setting with four or more agents, \citet{niemeyer2026optimal} show that deterministic optimal dominant-strategy IC mechanisms do not admit such a characterization. These mechanisms are extreme points of the set of BIC mechanisms and, therefore, uniquely optimal for some objective of the principal (if signal spaces are finite). Therefore, unless one has a reason to restrict the principal's objective, there is at least one relevant subclass of optimal BIC mechanisms that admits no economically meaningful characterization.

\appendix

\section{Proofs}\label{appendix:proofs}

\subsection{Incentive Compatibility: Proof of \headercref{Lemma}{{lemma:BIC_characterization}} and \headercref{Remark}{{remark:maxmin_value}}}\label{appendix:BIC_characterization}

\begin{step}
    If $x$ is BIC, then $x$ satisfies \eqref{eq:interim_indifference_constancy}.
\end{step}

The idea is to consider, for each agent $i$, the deviation where $i$ ignores their signal and reports randomly according to the marginal $\pi_{i}$. Using the zero-sum structure of the agents' payoffs, the fact that this deviation is unprofitable for both agents will imply that the deviation utility equals the truthtelling utility. Thus, all incentive constraints bind.

For all $i\in\lbrace \ell, r\rbrace$ and $s_{i}, s_{i}' \in S_{i}$, let $x_{i}(s_{i}, s_{i}') = \int_{S_{-i}} x(s_{i}', s_{-i}) \de \pi(s_{-i}\mid s_{i})$ be agent $i$'s expected decision given report $s_{i}'$ and signal $s_{i}$.
Given two distributions $\sigma_{\ell}\in\Delta(S_{\ell})$ and $\sigma_{r}\in\Delta(S_{r})$, let $x(\sigma_{\ell}, \sigma_{r}) = \int_{S} x \de \sigma_{\ell}\otimes \sigma_{r}$.
Finally, let $\bar{x} = \sup_{\sigma_{\ell}\in \Delta(S_{\ell})} \inf_{\sigma_{r}\in\Delta(S_{r})} x(\sigma_{\ell}, \sigma_{r})$.

In an intermediate step, we show $\bar{x}=\int_{S}x\de \pi$.
For arbitrary $\sigma_{\ell}\in\Delta(S_{\ell})$, BIC implies that agent $\ell$ prefers reporting truthfully to making a random report with distribution $\sigma_{\ell}$, i.e., $x_{\ell}(s_{\ell}, s_{\ell}) \geq 
\int_{S_{\ell}} \int_{S_{r}} x(s_{\ell}', s_{r}) \de\pi(s_{r}\mid s_{\ell})\de \sigma_{\ell}(s_{\ell}')
$ for $\pi_{\ell}$-almost all $s_{\ell}\in S_{\ell}$. 
By integrating this inequality with respect to $\pi_{\ell}$, we obtain:
\begin{align*}
    \int_{S} x \de \pi =  \int_{S_{\ell}} x_{\ell}(s_{\ell}, s_{\ell}) \de \pi_{\ell}(s_{\ell}) \geq 
    &\int_{S_{\ell}}\int_{S_{\ell}} \int_{S_{r}} x(s_{\ell}', s_{r}) \de\pi(s_{r}\mid s_{\ell})\de \sigma_{\ell}(s_{\ell}') \de\pi_{\ell}(s_{\ell}) \\
    =&\int_{S_{\ell}} \int_{S_{r}} x(s_{\ell}', s_{r}) \de \pi_{r}(s_{r}) \de\sigma_{\ell}(s_{\ell}') = x(\sigma_{\ell}, \pi_{r})
    .
\end{align*}
In particular, $\int_{S} x\de \pi \geq x(\sigma_{\ell}, \pi_{r})$ for all $\sigma_{\ell}\in\Delta(S_{\ell})$.
An analogous argument for agent $r$ shows $x(\pi_{\ell}, \sigma_{r}) \geq \int_{S}x\de \pi$ for all $\sigma_{r}\in\Delta(S_{r})$.
Summarizing, $x(\sigma_{\ell}, \pi_{r})\leq \int_{S} x\de \pi \leq x(\pi_{\ell}, \sigma_{r}) $  for all $\sigma_{\ell}\in\Delta(S_{\ell})$ and $\sigma_{r}\in\Delta(S_{r})$.
Thus,
\begin{equation*}
    \inf_{\sigma_{r}} \sup_{\sigma_{\ell}} x(\sigma_{\ell}, \sigma_{r}) 
    \leq \int_{S}x\de \pi
    \leq \sup_{\sigma_{\ell}} \inf_{\sigma_{r}} x(\sigma_{\ell}, \sigma_{r}) =\bar x
    .
\end{equation*}
The reverse inequalities hold in every two-player zero-sum game. Evaluating $x(\sigma_{\ell}, \pi_{r})\leq \int_{S} x\de \pi \leq x(\pi_{\ell}, \sigma_{r})$ at $\sigma_{\ell} = \pi_{\ell}$ and $\sigma_{r} = \pi_{r}$ also shows $\bar{x} = \int_{S}x \de \pi_{\ell}\otimes \pi_{r}$. In particular, $x(\sigma_{\ell}, \pi_{r}) \leq x(\pi_{\ell}, \pi_{r}) \leq x(\pi_{\ell}, \sigma_{r})$ for all $\sigma_{\ell}$ and $\sigma_{r}$, so $(\pi_{\ell}, \pi_{r})$ is a mixed Nash equilibrium of the auxiliary game, as claimed in \Cref{remark:maxmin_value}.

We now show that each agent $i$ is interim-indifferent between almost all reports $s_{i}'$ conditional on almost all signals $s_{i}$, i.e., $x_{i}(s_{i}, s_{i}) = x_{i}(s_{i}, s_{i}')$.
Take $i=\ell$, the proof for $r$ being analogous.
We have established $\int_{S}x\de\pi = \int_{S}x \de \pi_{\ell}\otimes \pi_{r}$, which rearranges to
$\int_{S_{\ell}}\int_{S_{\ell}} (x_{\ell}(s_{\ell}, s_{\ell}) - x_{\ell}(s_{\ell}, s_{\ell}')) \de \pi_{\ell}(s_{\ell}') \de \pi_{\ell}(s_{\ell}) = 0$. 
Since also $x_{\ell}(s_{\ell}, s_{\ell}) \geq x_{\ell}(s_{\ell}, s_{\ell}')$ holds for all $s_{\ell}'$ and $\pi_{\ell}$-almost all $s_{\ell}$ (by BIC), we conclude that $(\pi_{\ell}\otimes \pi_{\ell})$-almost all $(s_{\ell}, s_{\ell}')$ satisfy $x_{\ell}(s_{\ell}, s_{\ell}) = x_{\ell}(s_{\ell}, s_{\ell}')$.

It remains to show that the interim expectations $x_{i}(s_{i}, s_{i})$ equal $\bar{x}$.
Fix arbitrary $s_{\ell}\in S_{\ell}$ and $s_{r}\in S_{r}$ (outside the null sets on which BIC may fail).
By BIC, agent $\ell$ with signal $s_{\ell}$ prefers reporting truthfully to randomly reporting according to $\pi(\cdot\mid s_{r})$, i.e.,
\begin{equation*}
    x_{\ell}(s_{\ell}, s_{\ell}) \geq \int_{S_{\ell}} x_{\ell}(s_{\ell}, s_{\ell}') \de \pi(s_{\ell}'\mid s_{r}) = \int_{S_{\ell}}\int_{S_{r}} x(s_{\ell}', s_{r}') \de \pi(s_{r}'\mid s_{\ell}) \de\pi(s_{\ell}'\mid s_{r}).
\end{equation*}
Similarly, agent $r$ with signal $s_{r}$ prefers reporting truthfully to randomly reporting according to $\pi(\cdot\mid s_{\ell})$, i.e.,
\begin{equation*}
    x_{r}(s_{r}, s_{r}) \leq \int_{S_{r}} x_{r}(s_{r}, s_{r}') \de \pi(s_{r}'\mid s_{\ell}) = \int_{S_{r}}\int_{S_{\ell}} x(s_{\ell}', s_{r}')  \de \pi(s_{\ell}'\mid s_{r}) \de\pi(s_{r}'\mid s_{\ell}).
\end{equation*}
These two inequalities and Fubini's theorem imply $x_{r}(s_{r}, s_{r}) \leq x_{\ell}(s_{\ell}, s_{\ell})$.
Since $s_{\ell}$ and $s_{r}$ were arbitrary, and since both sides of this inequality integrate to $\int_{S}x \de \pi$, both sides must be almost everywhere equal to $\int_{S}x \de \pi$ and, hence, also equal to $\bar{x}$.

\begin{step}
    If $x$ satisfies \eqref{eq:interim_indifference_constancy}, then $x$ coincides with a BIC mechanism $\pi$-almost everywhere.
\end{step}

BIC requires that, for almost all signals, \emph{all} misreports are unprofitable, while \eqref{eq:interim_indifference_constancy} only concerns almost all misreports. The proof idea is simply to modify the mechanism on a null set.

In a technical step, we verify that, for each $i$, there are measurable sets $T_{i}$ (``true signals'') and $R_{i}$ (``reports'') such that $1 = \pi_{i}(T_{i}) = \pi_{i}(R_{i})$ and such that all $s_{i}\in T_{i}$ and $s_{i}'\in R_{i}$ satisfy $\E_{s_{-i}}[x(s_{i}', \bm{s_{-i}}) \mid s_{i}] = \E_{s}[x(\bm{s})]$; moreover, $\pi(\cdot\mid s_{i}) \ll \pi_{-i}$ for all $s_{i}\in T_{i}$.
Fix $i$.
\Cref{eq:interim_indifference_constancy} and Fubini's theorem imply that there is a set $R_{i}$ with $\pi_{i}(R_{i}) = 1$ such that for all $s_{i}' \in R_{i}$ there exists a set $T_{i}(s_{i}')$ such that $\pi_{i}(T_{i}(s_{i}')) = 1$ and $\E_{s_{-i}}[x(s_{i}', \bm{s_{-i}}) \mid s_{i}] = \E_{s}[x(\bm{s})]$ for all $s_{i}\in T_{i}(s_{i}')$. 
Note that the set $T_{i}(s_{i}')$ of valid $s_{i}$ may still depend on $s_{i}'$.
Consider the set of functions $\lbrace x(s_{i}', \cdot)\colon s_{i}'\in R_{i}\rbrace$, a subset of $L^{1}(S_{-i}, \pi_{-i})$.
With $S_{-i}$ Polish, $L^{1}(S_{-i}, \pi_{-i})$ is separable. 
Subsets of separable metric spaces are themselves separable.
Thus, $\lbrace x(s_{i}', \cdot)\colon s_{i}'\in R_{i}\rbrace$ admits a countable $L^{1}$-dense subset which we index by $\lbrace s_{i, n}'\rbrace_{n\in\N}\subseteq R_{i}$. 
For each $n$, since $s_{i, n}'\in R_{i}$, the set $T_{i, n} \overset{\mathrm{def}}{=} T_{i}(s_{i, n}')$ satisfies $\pi_{i}(T_{i, n}) = 1$ and $\E_{s_{-i}}[x(s_{i, n}', \bm{s_{-i}}) \mid s_{i}] = \E_{s}[x(\bm{s})]$ for all $s_{i}\in T_{i, n}$.
Next, since $\pi\ll\pi_{\ell}\otimes\pi_{r}$, there exists a set $\hat{T}_{i}$ with $\pi_{i}(\hat{T}_{i}) = 1$ such that $\pi(\cdot\vert s_{i})\ll\pi_{-i}$ for all $s_{i}\in \hat{T}_{i}$.
Now let $T_{i} = \hat{T}_{i} \cap \bigcap_{n\in\N} T_{i, n}$, so that $\pi_{i}(T_{i}) = 1$.
We claim $\E_{s_{-i}}[x(s_{i}', \bm{s_{-i}}) \mid s_{i}] = \E_{s}[x(\bm{s})]$ holds for all $s_{i}\in T_{i}$ and $s_{i}'\in R_{i}$.
Find a sequence $(s_{i, n_{k}}')_{k\in\N}$ such that $x(s_{i, n_{k}}', \cdot)$ $L^{1}$-converges to $x(s_{i}', \cdot)$. By definition of $T_{i}$, we have $s_{i}\in T_{i, n_{k}}$ for all $k$, meaning $\E_{s_{-i}}[x(s_{i, n_{k}}', \bm{s_{-i}}) \mid s_{i}] = \E_{s}[x(\bm{s})]$.
Since $\pi(\cdot\mid s_{i}) \ll \pi_{-i}$ (by $s_{i}\in\hat{T}_{i}$), since $x(s_{i, n_{k}}', \cdot)$ converges to $x(s_{i}', \cdot)$ in $L^{1}(S_{-i}, \pi_{-i})$, and since $x$ is bounded, we conclude $\E_{s_{-i}}[x(s_{i}', \bm{s_{-i}}) \mid s_{i}] = \E_{s}[x(\bm{s})]$.

We now describe the modification on a null set.
Define $y$ by $y(s) = x(s)$ for all $s\in R_{\ell}\times R_{r}$ and let $y$ equal the constant $\E_s[x(\bm s)]$ outside $R_\ell\times R_r$.
Since $\pi\ll\pi_{\ell}\otimes\pi_{r}$, and since the complement of $R_{\ell}\times R_{r}$ has $\pi_{\ell}\otimes\pi_{r}$-measure zero, $x$ and $y$ coincide $\pi$-almost everywhere.
We show that $y$ is BIC by showing that $\E_{s_{-i}}[y(s_{i}', \bm{s_{-i}}) \mid s_{i}] = \E_{s}[x(\bm{s})]$ for all $i$, all $s_{i}\in T_{i}$, and \emph{all} $s_{i}'\in S_{i}$; this suffices since $\pi_{i}(T_{i}) = 1$.
Take arbitrary $s_{i} \in T_{i}$.
The claim is immediate for $s_{i}'\notin R_{i}$ (since $y(s_{i}', \cdot)$ is then constantly equal to $\E_{s}[x(\bm{s})]$), so take $s_{i}'\in R_{i}$.
Since $\pi(\cdot\mid s_{i})\ll\pi_{-i}$, and since $S_{-i}\setminus R_{-i}$ has $\pi_{-i}$-measure zero, we get $\E_{s_{-i}}[y(s_{i}', \bm{s_{-i}}) \mid s_{i}] = \E_{s_{-i}}[x(s_{i}', \bm{s_{-i}}) \mid s_{i}] = \E_{s}[x(\bm{s})]$.
\qed

\subsection{The Role of Correlation: Proof of \headercref{Proposition}{{proposition:spanning_order}}}\label{appendix:role_of_correlation}

We prove a generalization of \Cref{proposition:spanning_order} that allows for abstract signal spaces.
We begin by defining the spanning order.
Let $\Pi$ be the set of joint distributions $\pi\in\Delta(S_{\ell}\times S_{r})$ that are absolutely continuous with respect to the product of their marginals, i.e., $\pi\ll\pi_{\ell}\otimes\pi_{r}$.
Given $\pi\in\Pi$ and $i\in\lbrace \ell, r\rbrace$, let the \emph{kernel} $\ker_{i}(\pi)$ be the set of bounded measurable $y\colon S_{-i}\to\R$ such that $\E_{s_{-i}}[y(\bm{s}_{-i})\mid s_{i}] = 0$ holds for $\pi_{i}$-almost all $s_{i}$.
Given $\pi, \pi' \in \Delta(S_{\ell}\times S_{r})$, we say $\pi$ \emph{spans} $\pi'$ if for both $i\in\lbrace \ell, r\rbrace$ it holds $\ker_{i}(\pi) \subseteq \ker_{i}(\pi')$.
Given $\pi\in\Delta(S_{\ell}\times S_{r})$ and $i$, we say $\pi$ is \emph{boundedly complete} with respect to $i$ if every function in $\ker_{i}(\pi)$ is $\pi_{-i}$-almost everywhere equal to $0$.\footnote{This notion of bounded completeness is a variant of the notion found in the statistics and econometrics literature in the context of identification. There, one is typically interested in the distribution of a random variable, call it $\bm{s}_{r}$, conditional on the distribution of another random variable $\bm{s}_{\ell}$. Meanwhile, we are interested also in the distribution of $\bm{s}_{\ell}$ conditional on $\bm{s}_{r}$ as there is nothing special about either agent.}

When both signal spaces, $S_{\ell}$ and $S_{r}$, are finite, this notion of the spanning order coincides with the one given in the main text, and bounded completeness of $\pi$ (with respect to at least one $i$) means that $\pi$, regarded as a matrix, has full column- or row-rank.
Accordingly, the following proposition implies \Cref{proposition:spanning_order}.
\begin{proposition}\label{proposition:general_spanning_order}
    For all $\pi, \pi'\in \Pi$, if $\pi$ spans $\pi'$, then for every $x\in \mathcal{X}(\pi)$ there exists $x'\in \mathcal{X}(\pi')$ that agrees with $x$ at $\pi'$-almost all points.\footnote{By adapting Step 2 from the proof of \Cref{lemma:BIC_characterization}, one can choose $x'$ to coincide with $x$ at $\pi_{\ell}\otimes\pi_{r}$-almost all points.
    }
    Consequently, if a mechanism $x$ is BIC under some $\pi\in\Pi$, then there exists $x'\in \mathcal{X}(\pi_{\ell}\otimes\pi_{r})$ that agrees with $x$ at $\pi_{\ell}\otimes\pi_{r}$-almost all points.
    Moreover, for all $\pi\in\Pi$, if $\pi$ is boundedly complete with respect to at least one $i\in\lbrace \ell, r\rbrace$, then all $x\in\mathcal{X}(\pi)$ are constant $\pi$-almost everywhere; the converse holds if $\pi$ and $\pi_{\ell}\otimes\pi_{r}$ are mutually absolutely continuous. 
\end{proposition}

\begin{proof}[Proof of \Cref{proposition:general_spanning_order}]
    Let $x\in\mathcal{X}(\pi)$. For $\bar{x}$ defined in \Cref{remark:maxmin_value}, \Cref{lemma:BIC_characterization} asserts that all $i\in\lbrace \ell, r\rbrace$ and $\pi_{i}\otimes\pi_{i}$-almost all $s_{i}, s_{i}'\in S_{i}$ satisfy $\E_{s_{-i}}[x(s_{i}', \bm{s_{-i}}) - \bar{x} \mid s_{i}] = 0$, i.e., the function $x(s_{i}', \cdot) - \bar{x}$ is in $\ker_{i}(\pi)$.
    Since $\pi$ spans $\pi'$, also for all $i$ and $\pi_{i}$-almost all $s_{i}'$ the function  $x(s_{i}', \cdot) - \bar{x}$ is in $\ker_{i}(\pi')$.
    Using again that $\pi$ spans $\pi'$, one can show $\pi_{i}'\ll \pi_{i}$.\footnote{Take $E$ such that $\pi_{i}(E) = 0$, meaning the indicator function for $E$ is in $\ker_{-i}(\pi)$. Since $\pi$ spans $\pi'$, the indicator is also in $\ker_{-i}(\pi')$, so that $\pi_{i}'(E) = 0$.}
    In particular, for all $i$ and $\pi_{i}'$-almost all $s_{i}'$ the function $x(s_{i}', \cdot) - \bar{x}$ is in $\ker_{i}(\pi')$.
    According to \Cref{lemma:BIC_characterization}, there exists $x'\in \mathcal{X}(\pi')$ that agrees with $x$ at $\pi'$-almost all points.

    Now suppose $\pi\in\Pi$ is boundedly complete with respect to some $i$.
    Without loss, let $i = \ell$.
    Take $x\in\mathcal{X}(\pi)$.
    Thus, $x(s_{\ell}', \cdot) - \bar{x} \in \ker_{\ell}(\pi)$ for $\pi_{\ell}$-almost all $s_{\ell}'$.
    By bounded completeness, $x(s_{\ell}', \cdot) = \bar{x}$ for $\pi_{\ell}$-almost all $s_{\ell}'$.
    That is, $\pi_{\ell}\otimes\pi_{r}$-almost all $s_{\ell}'$ and $s_{r}'$ satisfy $x(s_{\ell}', s_{r}') = \bar{x}$.
    Since $\pi$ is absolutely continuous with respect to $\pi_{\ell}\otimes\pi_{r}$, the equation $x = \bar{x}$ also holds $\pi$-almost everywhere.

    Finally, suppose $\pi\in\Pi$ is boundedly complete with respect to neither $\ell$ nor $r$.
    Thus, for each $i$ there is a bounded $f_{-i}\in\ker_{i}(\pi)$ such that $f_{-i}$ is not $\pi_{-i}$-almost everywhere equal to $0$.
    By boundedness, we can find $\varepsilon > 0$ such that $x$ defined for all $s$ by $x(s) = 1/2 + \varepsilon f_{\ell}(s_{\ell}) f_{r}(s_{r})$ maps to $[0, 1]$, i.e., $x$ is a mechanism.
    For each $i$ and $\pi_{i}$-almost all $s_{i}$ and all $s_{i}'$, it holds $\E_{s_{-i}}[x(s_{i}', \bm{s_{-i}})\mid s_{i}] = 1/2 + \varepsilon f_{i}(s_{i}') \E_{s_{-i}}[f_{-i}(\bm{s_{-i}})\mid s_{i}] = 1/2$ since $f_{-i}\in\ker_{i}(\pi)$.
    Thus, $x$ is BIC.
    Finally, since for each $i$ the function $f_{i}$ is not $\pi_{i}$-almost everywhere equal to $0$, the product $f_{\ell} f_{r}$ is non-zero on a set of strictly positive $\pi_{\ell}\otimes\pi_{r}$-measure. 
    Since $f_{-i}\in\ker_{i}(\pi)$ for each $i$, 
    we have $\E_{s\sim\pi_{\ell}\otimes\pi_{r}}[f_{\ell}(\bm{s_{\ell}}) f_{r}(\bm{s_{r}})] = 0$ and can, therefore, conclude that $f_{\ell} f_{r}$ is not $\pi_{\ell}\otimes\pi_{r}$-almost everywhere constant. Since $\pi$ and $\pi_{\ell}\otimes\pi_{r}$ are mutually absolutely continuous, $f_{\ell} f_{r}$ is not $\pi$-almost everywhere constant either. Thus, $x$ is not $\pi$-almost everywhere constant.
\end{proof}

\subsection{The Role of Interdependence: Proof of \headercref{Proposition}{{proposition:profitability_independence}}}\label{appendix:proof_profitability_independence}

We begin with a piece of notation and a useful observation.
Given a function $f\in L^{1}(S_{\ell}\times S_{r}, \pi)$, for each $i\in\lbrace \ell, r\rbrace$ and $s_{i}\in S_{i}$, let $f_{i}(s_{i}) = \E_{s_{-i}}[f(s_{i}, \bm{s_{-i}})]$.
Then, let $\bar{f} = f - f_{\ell} - f_{r} + \E_{s}[f]$.
Notice that, by construction, $\E_{s}[\bar{f}] = 0$ and that $f-\bar{f}$ is additive.

We now observe that for all $f\in L^{1}(S_{\ell}\times S_{r}, \pi)$ and all additive $a = a_{\ell} + a_{r}$ such that $\bar{f}a_{\ell}$ and $\bar{f}a_{r}$ are integrable (e.g., if $f$ or $a$ is bounded), it holds $\E_{s}[\bar{f}a] = 0$.
Additivity means there are $a_{\ell}$ and $a_{r}$ such that $a(s) = a_{\ell}(s_{\ell}) + a_{r}(s_{r})$ for $\pi_{\ell}\otimes\pi_{r}$-almost all $s$.
Then, for each $s_{\ell}$, we have $\E_{s_{r}}[(f(s_{\ell}, \bm{s_{r}}) - f_{\ell}(s_{\ell})) a_{\ell}(s_{\ell})] = 0$ by definition of $f_{\ell}$, and we have $\E_{s_{r}}[(f_{r}(\bm{s_{r}}) - \E_{s}[f])a_{\ell}(s_{\ell})] = 0$ by definition of $f_{r}$. Thus, $\E_{s}[\bar{f} a_{\ell}] = 0$. A similar argument shows $\E_{s}[\bar{f} a_{r}] = 0$.

Now let $\bar{v}$ be defined as above with $f = v$; since $\E_{s}[v] = 0$ by assumption, we have $\bar{v} = v - v_{\ell} - v_{r}$.
Let $\eta = \sgn(\bar{v})$.
Finally, let $\bar{\eta}$ be defined as above with $f=\eta$.
Since $\eta$ takes values in $\lbrace -1, 0, 1\rbrace$, the function $\bar{\eta}$ is bounded.
Thus, there exists $\varepsilon > 0$ such that $x = \tfrac{1}{2} + \varepsilon\bar{\eta}$ is a mechanism, i.e., maps to $[0, 1]$.
We claim $x$ is a profitable BIC mechanism.
For later reference, note $x - \E_{s}[x] = \varepsilon(\bar{\eta} - \E_{s}[\bar{\eta}]) = \varepsilon\bar{\eta}$.

To establish BIC, we show that for all $i$ and $s_{i}$ it holds $\E_{s_{-i}}[x(s_{i}, \bm{s_{-i}}) - \E_{s}[x]] = 0$.
(Notice that since signals are independent, $\E_{s_{-i}}[x(s_{i}, \bm{s_{-i}})]$ is indeed agent $i$'s expected probability from reporting $s_{i}$ and conditional on any private signal).
Recall $x - \E_{s}[x] = \varepsilon\bar{\eta}$.
Thus, we show $\E_{s_{-i}}[\bar{\eta}(s_{i}, \bm{s_{-i}})] = 0$ for all $i$ and $s_{i}$.
Indeed, for every $s_{i}$, the definitions of $\bar{\eta}$, $\eta_{i}$, and $\eta_{-i}$ give $\E_{s_{-i}}[\bar{\eta}(s_{i}, \bm{s_{-i}})] = \eta_{i}(s_{i}) - \eta_{i}(s_{i}) - \E_{s}[\eta] + \E_{s}[\eta] = 0$.
Thus, $x$ is BIC.

It remains to show that $x$ is profitable, i.e., $\E_{s}[xv] > \max\lbrace 0, \E_{s}[v]\rbrace$. Since $\E_{s}[v] = 0$ by assumption, this is equivalent to $\E_{s}[(x - \E_{s}[x]) v] > 0$.
Recalling $x - \E_{s}[x] = \varepsilon\bar{\eta}$, we have to show $\E_{s}[\bar{\eta} v] > 0$.
Recall the observation: $\E[\bar{f} a]=0$ for all $f$ and additive $a$ such that $f$ or $a$ is bounded.
We now use this observation twice, first with $f=\eta$ and $a = v - \bar{v}$, and then with $f = v$ and $a = \eta - \bar{\eta}$, which yields $\E_{s}[\bar{\eta} v] = \E_{s}[\bar{\eta} \bar{v}] = \E_{s}[\eta \bar{v}]$.
Recalling also the definition $\eta = \sgn(\bar{v})$, the expectation of interest becomes $\E_{s}[\eta \bar{v}] = \E_{s}[\sgn(\bar{v}) \bar{v}] = \E_{s}[\vert\bar{v}\vert]$.
By assumption, $v$ is non-additive, so $\bar{v} = v - v_{\ell} - v_{r} \neq 0$, proving $\E_{s}[\bar{\eta} v]=\E_{s}[\vert\bar{v}\vert] > 0$, as desired.
\qed

\subsection{Optimal Mechanisms: Proof of \headercref{Theorem}{{theorem:optimal_mech}}}\label{appendix:proof_optimal_mech}
Let $F_{i}$ be the set of $f_{i}\in L^{1}(S_{i}, \pi_{i})$ such that $\E_{s_{i}}[f_{i}(\bm{s_{i}})] = 0$.
Let $F$ be the set of $f\colon S\to\R$ for which there are $f_{\ell}\in F_{\ell}$ and $f_{r}\in F_{r}$ such that $f = f_{\ell} + f_{r}$.
Let $X$ be the set of all (not necessarily BIC) mechanisms, i.e., measurable $x\colon S\to [0, 1]$, where we identify mechanisms that agree $\pi$-almost everywhere (which is without loss of generality by \Cref{lemma:BIC_characterization}).

Under independence, a mechanism $x\in X$ is BIC for agent $i\in\lbrace \ell, r\rbrace$ if and only if for $\pi_{i}$-almost all $s_{i}\in S_{i}$ it holds $\E_{s_{-i}}[x(s_{i}, \bm{s_{-i}}) - \E_{s}[x]] = 0$.
Equivalently, for all $f_{i}\in L^{1}(S_{i}, \pi_{i})$ it holds $ 0 = \E_{s}\left[f_{i}(\bm{s_{i}})\left(x(\bm{s_{i}}, \bm{s_{-i}}) - \E_{s}[x]\right)\right]$.
We can write the right side as
\begin{align*}
    \E_{s}\left[f_{i}(\bm{s_{i}})\left(x(\bm{s}) - \E_{s}[x]\right)\right] &= \E_{s}\left[(f_{i}(\bm{s_{i}}) - \E_{s_{i}}[f_{i}])\left(x(\bm{s}) - \E_{s}[x]\right)\right] \\ &= \E_{s}\left[(f_{i}(\bm{s_{i}}) - \E_{s_{i}}[f_{i}])x(\bm{s})\right].
\end{align*}
Recall that $F_{i}$ is the set of mean-zero $f_{i}\in L^{1}(S_{i}, \pi_{i})$.
Thus, $x$ is BIC for $i$ if and only if, for all $f_{i}\in F_{i}$, it holds $ 0 = \E_{s}\left[f_{i}(\bm{s_{i}}) x(\bm{s})\right]$.

It follows that, if $x$ is BIC, then all $f\in F$ satisfy $\E_{s}[v x] = \E_{s}[(v + f) x]$. Meanwhile, if $x$ is not BIC, then $\inf_{f\in F} \E_{s}[f x] = -\infty$.
Thus, the principal's problem is to maximize $\inf_{f\in F} \E_{s}[(v + f) x]$
across $x\in X$.

We next use Sion's Minimax Theorem to obtain a dual problem.
Note $X$ is a weak-$\ast$ closed subset of the unit ball in $L^{\infty}(S, \pi)$, where $L^{\infty}(S, \pi)$ in turn is the norm dual of $L^{1}(S, \pi)$.
Alaoglu's theorem implies that $X$ is weak-$\ast$ compact.
Since $v + f \in L^{1}(S, \pi)$ for all $f\in F$, we have that $\E_{s}[(v + f) x]$ is weak-$\ast$ continuous in $x$. Similarly, $\E_{s}[(v + f) x]$ is weakly continuous in $f$.
We also note that this expectation is linear in $x$ and affine in $f$, and that $X$ and $F$ are both convex. 
Hence, Sion's Minimax Theorem implies
\begin{equation*}
    \max_{x\in X}\inf_{f\in F} \E_{s}[(v + f) x] = \inf_{f\in F} \max_{x\in X} \E_{s}[(v + f) x].
\end{equation*}
The right side simplifies to $\inf_{f\in F} \E_{s}[\max\lbrace 0, v + f\rbrace]$.
(Since $X$ is weak-$\ast$ compact and both $x\mapsto\E_{s}[(v + f) x]$ and $x\mapsto\inf_{f\in F}\E_{s}[(v + f) x]$ are weak-$\ast$ upper semicontinuous, the maxima on both sides are attained.)

We show further ahead that $\inf_{f\in F} \E_{s}[\max\lbrace 0, v + f\rbrace]$ is attained at some $f\in F$.
Taking this as given, we complete the proof.
Put $v^{\ast} = v + f$ for the candidate virtual value.
Let $x\in X$ be an optimal BIC mechanism (existence follows from the above continuity-compactness observations).
By BIC, $\E_{s}[vx] = \E_{s}[v^{\ast} x] \leq  \E_{s}[\max\lbrace 0, v^{\ast}\rbrace]$.
By the minimax equality, also $\E_{s}[vx] = \E_{s}[\max\lbrace 0, v^{\ast}\rbrace]$.
Thus, for $\pi$-almost all $s$, it holds $x(s) = 1$ if $v^{\ast}(s) > 0$, and $x(s) = 0$ if $v^{\ast}(s) < 0$.
On the set where $v^{\ast} = 0$, the payoff $v$ is additive (since $v = v^{\ast} - f$ and $f$ is additive).
By assumption, $v$ is nowhere additive.
Thus, $x$ is for almost all $s$ determined by $x(s) = \mathbbm{1}_{v^{\ast}(s) \geq 0}$.

It remains to argue that $\inf_{f\in F} \E_{s}[\max\lbrace 0, v + f\rbrace]$ is attained.
The idea is to associate $f\in F$ with signed Radon measures to leverage the weak-$\ast$ sequential compactness properties of that space. 
This idea follows \citet{korman2015dual}.

For all $f\in F$, we have $2\E_{s}[\max\lbrace 0, v + f\rbrace] = \E_{s}[\vert v + f\vert] + \E_{s}[v + f]$.
Since $\E_{s}[f] = 0$ by definition of $F$, it suffices to minimize $\E_{s}[\vert v + f\vert]$.
Let $M = \inf_{f\in F}\E_{s}[\vert v + f\vert]$, which plainly exists since the infimum is weakly positive.
Let $(f_{n})_{n\in\mathbb{N}}$ be a sequence in $F$ along which $\E_{s}[\vert v + f_{n}\vert]$ converges to $M$.
By definition, for each $n\in\N$, there are $f_{n, \ell}\in F_{\ell}$ and $f_{n, r}\in F_{r}$ such that $f_{n} = f_{n, \ell} + f_{n, r}$.

We begin by deriving $L^{1}$-bounds on $(f_{n, \ell})_{n}$ and $(f_{n, r})_{n}$. 
For each $i$, the definition of $F_{-i}$ yields $\E_{s_{-i}}[f_{n, -i}] = 0$, so that for all $s_{i}$ it holds $f_{n, i}(s_{i}) = \E_{s_{-i}}[f_{n}(s_{i}, \bm{s_{-i}})]$.
It follows
\begin{equation*}
    \E_{s_{i}}[\vert f_{n, i}\vert] \leq \E_{s}[\vert f_{n}\vert] = \E_{s}[\vert v + f_{n} - v\vert] \leq \E_{s}[\vert v + f_{n}\vert] + \E_{s}[\vert v\vert].
\end{equation*}
Since $\E_{s}[\vert v + f_{n}\vert]$ converges to $M$, we have that $\E_{s_{i}}[\vert f_{n, i}\vert]$ is bounded across $n$ and $i$.

Let $\mathcal{M}(S_{i})$ denote the set of signed Radon measures on $S_{i}$.
Consider the signed Radon measure $U_{n, i} = f_{n, i} \pi_{i}$.\footnote{That is, $U_{n, i}(E) = \int_{E} f_{n, i}\de\pi_{i}$ for each measurable $E\subseteq S_{i}$. Since $S_{i}$ is a compact metric space, every finite signed Borel measure on $S_{i}$ is Radon.}
Note, $U_{n, i}(S_{i}) = \E_{s_{i}}[f_{n, i}] = 0$. 
Further, let $T_{n} = v \pi_{\ell}\otimes \pi_{r} + U_{n, \ell} \otimes \pi_{r} + \pi_{\ell}\otimes U_{n, r}$.
In this notation, the objective $\E_{s}[\vert v + f_{n}\vert] = \int_{S}\vert v + f_{n, \ell} + f_{n, r}\vert \de \pi_{\ell}\otimes \pi_{r}$ simply equals the total variation norm $\Vert T_{n}\Vert_{TV}$ of $T_{n}$.

For each $i\in\lbrace \ell, r\rbrace$, the total variation norm of $U_{n, i}$ is $\Vert U_{n, i}\Vert_{TV} = \int_{S_{i}} \vert f_{n, i}\vert \de \pi_{i}$, which is simply $\E_{s_{i}}[\vert f_{n, i}\vert]$.
In particular, for each $i$ the sequence $(U_{n, i})_{n\in\N}$ is bounded in the total variation norm.
The space $\mathcal{M}(S_{i})$ with the total variation norm is the dual of $C(S_{i})$.
Thus, the bounded balls in $\mathcal{M}(S_{i})$ are weak-$\ast$ compact (Alaoglu's theorem).
Moreover, since $S_{i}$ is compact and metrizable, $C(S_{i})$ is separable (e.g. \citet[Lemma 3.99]{aliprantis2006infinite}), and so the bounded balls in $\mathcal{M}(S_{i})$ are weak-${\ast}$ metrizable (e.g. \citet[Theorem 6.30]{aliprantis2006infinite}).
In particular, using sequential compactness and by possibly passing to a subsequence, for each $i\in\lbrace \ell, r\rbrace$ the sequence $(U_{n, i})_{n\in\N}$ weak-$\ast$ converges to some $U_{i}\in \mathcal{M}(S_{i})$.
Consequently, $T_{n}$ weak-$\ast$ converges to $T = v \pi_{\ell}\otimes \pi_{r} + U_{\ell} \otimes \pi_{r} + \pi_{\ell}\otimes U_{r}$.
Since the total variation norm is lower semicontinuous, and since $\E_{s}[\vert v + f_{n}\vert] = \Vert T_{n}\Vert_{TV}$ converges to $M$, we get $\Vert T\Vert_{TV} \leq M$.

To complete the proof, we construct $f\in F$ with $\E_{s}[\vert v + f\vert] \leq \Vert T\Vert_{TV}$ from the limit measures $U_\ell$ and $U_r$; since $\Vert T\Vert_{TV}\leq M$, such an $f$ attains the infimum $M$.
Using Lebesgue's decomposition theorem, there exist $f_{i}\in L^{1}(S_{i}, \pi_{i})$ and $U_{i}'\in \mathcal{M}(S_{i})$ such that $U_{i} = f_{i} \pi_{i} + U_{i}'$, where $U_{i}'$ and $\pi_{i}$ are mutually singular.
For each $i$, since $U_{n, i}(S_{i}) = 0$ for all $n$, we have $U_{i}(S_{i}) = 0$.
Thus, $U_{\ell}'(S_{\ell}) + U_{r}'(S_{r}) = - \E_{s}[f_{\ell} + f_{r}]$.
Further, plugging the decomposition into $T$, we get
\begin{equation*}
    T = (v + f_{\ell} + f_{r})\pi_{\ell}\otimes \pi_{r} + U_{\ell}' \otimes \pi_{r} + \pi_{\ell}\otimes U_{r}'.
\end{equation*}
Since $U_{i}'$ and $\pi_{i}$ are mutually singular for each $i$, we have
$\Vert T\Vert_{TV} = \int_{S}\vert v + f_{\ell} + f_{r}\vert \de \pi_{\ell}\otimes \pi_{r} + \Vert U_{\ell}' \otimes \pi_{r} + \pi_{\ell}\otimes U_{r}'\Vert_{TV}$.

Now define $f = f_{\ell} - \E_{s_{\ell}}[f_{\ell}] + f_{r} - \E_{s_{r}}[f_{r}]$.
Since both $f_{\ell} - \E_{s_{\ell}}[f_{\ell}]$ and $f_{r} - \E_{s_{r}}[f_{r}]$ have mean-zero, we have $f\in F$.
It remains to verify $\mathbb{E}_s[|v+f|]\le\|T\|_{TV}$.
We have
\begin{align*}
    \E_{s}[\vert v + f\vert] 
    &= \int_{S}\vert v + f_{\ell} + f_{r} - \E_{s}[f_{\ell}+f_{r}] \vert \de \pi_{\ell}\otimes\pi_{r} \\
    &\leq  \int_{S}\vert v + f_{\ell} + f_{r}\vert \de \pi_{\ell}\otimes\pi_{r} +  \underbrace{\vert\E_{s}[f_{\ell}+f_{r}]\vert}_{=\vert U_{\ell}'(S_{\ell}) + U_{r}'(S_{r})\vert} \\
    &\leq  \int_{S}\vert v + f_{\ell} + f_{r}\vert \de \pi_{\ell}\otimes\pi_{r} + \Vert U_{\ell}' \otimes \pi_{r} + \pi_{\ell}\otimes U_{r}'\Vert_{TV} = \Vert T\Vert_{TV}. \tag*{\qedsymbol}
\end{align*}

\subsection{A sufficient condition for the payoff to be nowhere additive}\label{appendix:sufficient_for_nonadditive}

We provide a sufficient condition for the principal's payoff to be nowhere additive in the sense of \hyperref[assumption:regularity]{Regularity}.

Let $S_{\ell} = S_{r} = [0, 1]^{m}$ for some common $m\in\N$. (That $m$ is common has no bite, as one could add payoff-irrelevant dimensions.)
Each signal profile thus corresponds to a vector $s = (s^{1}, \ldots, s^{2m})$, where the first $m$ entries correspond to agent $\ell$'s $m$-dimensional signal and the last $m$ entries correspond to agent $r$'s $m$-dimensional signal.
Our sufficient condition is that, at almost every signal profile $s$, at least one dimension $i$ of $\ell$'s signal and one dimension $j$ of $r$'s signal are interdependent, in the sense that the cross-partial derivative of the principal's payoff $v$ is non-zero, $\partial_{i}\partial_{j} v(s) \neq 0$.
\begin{lemma}\label{lemma:sufficient_for_nonadditive}
    Suppose the agents' signals are independent and each admits a
    density on $[0, 1]^{m}$.
    Suppose that the principal's payoff $v$ is twice continuously differentiable.
    If for $\Leb$-almost all $s\in [0, 1]^{2m}$ there exist $i\in\lbrace 1, \ldots, m\rbrace$ and $j\in\lbrace m+1, \ldots, 2m\rbrace$ such that $\partial_{i}\partial_{j} v(s) \neq 0$, then $v$ is nowhere additive.
\end{lemma}

\begin{proof}[Proof of \Cref{lemma:sufficient_for_nonadditive}]
    Towards a contradiction, suppose not.
    Then, there are integrable $w_{\ell}$ and $w_{r}$ such that $v(s) = w_{\ell}(s_{\ell}) + w_{r}(s_{r})$ for all $s$ in a set $A$ of non-zero Lebesgue measure. (Here, we recall that nowhere additivity is defined with respect to $\pi = \pi_{\ell}\otimes\pi_{r}$ by independence, and then we use that each of $\pi_{\ell}$ and $\pi_{r}$ admits a density to pass to Lebesgue measure.)
    Let $B_{+}$ be the set of points $s$ for which there exists $i\in\lbrace 1, \ldots, m\rbrace$ and $j\in\lbrace m+1, \ldots, 2m\rbrace$ such that the cross-partial derivative $\partial_{i}\partial_{j} v(s)$ is strictly positive.
    Let $B_{-}$ analogously be the set of points where at least one such cross-partial is strictly negative.
    By assumption, $\Leb(S\setminus (B_{-}\cup B_{+})) = 0$.
    Since $\Leb(A) > 0$, one of $A\cap B_{-}$ and $A\cap B_{+}$ has non-zero measure; suppose $\Leb(A\cap B_{+}) > 0$, the other case being analogous.
    By inner regularity of Lebesgue measure, there is a closed subset $A'\subseteq A\cap B_{+}$ such that $\Leb(A') > 0$.
    Given $s\in\R^{2m}$, define the translate $A' + s = \lbrace a + s\colon a\in A'\rbrace$.
    For some $\varepsilon > 0$ to be chosen, we shall consider translates of $A'$ in various directions.
    For all $i\in\lbrace 1, \ldots, 2m\rbrace$, let $e_{i}$ be the unit vector for dimension $i$. 
    For convenience, let $e_{0} = 0$ be the zero-vector.
    Then, as $i$ and $j$ both range over $\lbrace 0, \ldots, 2m\rbrace$, the translates $A' - \varepsilon e_i - \varepsilon e_j$ include $A'$ itself and its translates along one or two coordinate directions.
    Let
    \begin{equation*}
        A'' = \bigcap_{i,j\in\lbrace 0, 1, \ldots, 2m\rbrace} (A' - \varepsilon e_{i} - \varepsilon e_{j})
    \end{equation*}
    be the intersection of all such translates.
    Since $\Leb(A') > 0$, for all sufficiently small $\varepsilon$ it also holds $\Leb(A'') > 0$.\footnote{This fact follows from the  arguments behind Steinhaus' theorem.
    Fix finite $B\subseteq \R^{2m}$. It holds $\Leb(\cap_{b\in B} (A' + \varepsilon b)) \geq \Leb(A') - \sum_{b\in B} \Leb(A' \setminus (A' + \varepsilon b))$. It is a standard fact that $\lim_{\varepsilon\to 0} \Leb(A' \setminus (A' + \varepsilon b)) = 0$ holds. Thus, $\lim_{\varepsilon\to 0}\Leb(\cap_{b\in B} (A' + \varepsilon b)) > 0$.}
    In particular, $A''$ is non-empty for all sufficiently small $\varepsilon>0$.
    For each such $\varepsilon > 0$, let $s^{\varepsilon}\in A''$.
    Note $A'' \subseteq A'$.
    By compactness of $S$, letting $\varepsilon\to 0$, we may assume $s^{\varepsilon}$ converges to some $s$. Since $A'$ is closed, also $s\in A'$. Recall $A' \subseteq A\cap B_{+}$.
    By definition of $B_{+}$, there exists $i\in\lbrace 1, \ldots, m\rbrace$ and $j\in\lbrace m+1, \ldots, 2m\rbrace$ such that $\partial_{j}\partial_{i} v(s) > 0$.
    By the smoothness assumptions on $v$, the same inequality holds on a neighborhood $O$ of $s$.
    Take $\varepsilon > 0$ sufficiently small such that the rectangle with vertices $\lbrace s^{\varepsilon}, s^{\varepsilon} + \varepsilon e_{i}, s^{\varepsilon} + \varepsilon e_{j}, s^{\varepsilon} + \varepsilon e_i + \varepsilon e_j\rbrace$ is contained in $O$.
    Now, since $s^{\varepsilon}\in A''$, the definition of $A''$ also implies that $s^{\varepsilon}$, $s^{\varepsilon} + \varepsilon e_{i}$, $s^{\varepsilon} + \varepsilon e_{j}$, and $s^{\varepsilon} + \varepsilon e_i + \varepsilon e_j$ are all in $A'$ and, hence, in $A$.
    In particular, $v=w_\ell +w_r$ at each of these four points.
    Consequently, 
    \begin{equation*}
        0 = v(s^\varepsilon+\varepsilon e_i+\varepsilon e_j)-v(s^\varepsilon+\varepsilon e_i)-v(s^\varepsilon+\varepsilon e_j)+v(s^\varepsilon)=\int_{s_{j}^{\varepsilon}}^{s_{j}^{\varepsilon} + \varepsilon} \int_{s_{i}^{\varepsilon}}^{s_{i}^{\varepsilon} + \varepsilon} \partial_{j}\partial_{i} v(\tilde{s}_{i}, \tilde{s}_{j}, s_{-ij}^{\varepsilon}) \de \tilde{s}_{i}\de\tilde{s}_{j}
        .
    \end{equation*}
    But since the rectangle with vertices $\lbrace s^{\varepsilon}, s^{\varepsilon} + \varepsilon e_{i}, s^{\varepsilon} + \varepsilon e_{j}, s^{\varepsilon} + \varepsilon e_i + \varepsilon e_j\rbrace$ is contained in $O$, the integrand is strictly positive at all points. Contradiction. 
\end{proof}

\subsection{Communication without Commitment: Proof of \headercref{Theorem}{{theorem:no_commitment_profitable}}}\label{appendix:proof_no_commitment_profitable}
    We construct a PBE of an indirect mechanism in which the two agents $\ell$ and $r$ simultaneously send one of two messages, $m_{\ell},m_r\in\lbrace -1, 1\rbrace$.
    We write the probability of sending message $m_{i}$ as $\sigma_{i}(m_{i}\mid s_{i}) = (1 + m_{i}\tau_{i}(s_{i})) / 2$, where $\tau_{i}$ maps to $[-1, 1]$.
    In the candidate PBE, the principal takes decision $L$ if and only if the agents send the same message, $m_{\ell} = m_{r}$.
    
    We next describe the key properties of $(\tau_{\ell}, \tau_{r})$ that we use for the construction.
    First, we will choose each $\tau_{i}$ such that 
    \begin{equation}\label{eq:tau_requirement}
        \E_{s_{i}}[\tau_i(\bm{s_{i}})]=\E_{s}[v(\bm{s_{\ell}},\bm{s_{r}})\tau_i(\bm{s_{i}})]=0.
    \end{equation}
    Let $\mathcal{T}_{i}$ denote the set of bounded measurable $\tau_{i}\colon S_{i}\to \mathbb{R}$ satisfying \eqref{eq:tau_requirement}. 
    Second, we will choose $\tau_{\ell}$ and $\tau_{r}$ to jointly satisfy
    \begin{equation}\label{eq:tau_profitability}
        \mathbb{E}_{s}[v(\bm{s}) \tau_{\ell}(\bm{s}_{\ell}) \tau_{r} (\bm{s}_{r})] > 0.
    \end{equation}
    For a moment, let us assume $\tau_{\ell}$ and $\tau_{r}$ with the above properties exist. We argue that the behavior described above constitutes a profitable PBE.
    
    Consider the perspective of an agent $i$.
    Since signals are independent, the message $\bm{m_{-i}}$ is independent of $s_{i}$. Since $\E_{s_{-i}}[\tau_{-i}(\bm{s_{-i}})]=0$, from agent $i$'s perspective $\bm{m_{-i}}$ takes each value ($-1$ or $1$) with ex-ante probability $1/2$. Since the principal's decision only depends on whether the agents' messages match, for each message that agent $i$ could send, agent $i$ expects decision $L$ with probability $1/2$. In particular, $\sigma_{i}$ is a best reply.
    
    Now consider the principal's perspective.
    Each message profile $(m_{\ell}, m_{r}) \in \lbrace -1, 1\rbrace^{2}$ occurs with ex-ante probability $1/4$.
  Thus, given messages $(m_{\ell}, m_{r})$, the principal's posterior for $(\bm{s_{\ell}}, \bm{s_{r}})$ has density $(1 + m_{\ell}\tau_{\ell}(s_{\ell}))(1 + m_{r}\tau_{r}(s_{r}))$ with respect to $\pi_{\ell}\otimes\pi_{r}$.
    Hence, the principal's posterior expected payoff equals:
    \begin{align*}
        \E_{s} [v(\bm{s})\mid m_{\ell}, m_{r}] =& \E_{s}[v(\bm{s})] + m_{\ell} \E_{s} [v(\bm{s}) \tau_{\ell}(\bm{s_{\ell}})] + m_{r} \E_{s} [v(\bm{s}) \tau_{r}(\bm{s_{r}})] \\
        & + m_{\ell} m_{r} \E_{s} [v(\bm{s}) \tau_{\ell}(\bm{s_{\ell}}) \tau_{r}(\bm{s_{r}})] \\
        =& m_{\ell} m_{r} \E_{s} [v(\bm{s}) \tau_{\ell}(\bm{s_{\ell}}) \tau_{r}(\bm{s_{r}})],
    \end{align*}
    where we have used \eqref{eq:tau_requirement} and that the principal is unbiased.
    In view of \eqref{eq:tau_profitability},
    it follows that the principal indeed has a (strict) best reply of taking decision $L$ if and only if $m_{\ell} = m_{r}$.
    Thus, the described strategies form a PBE.
    By the same argument, \eqref{eq:tau_profitability} implies that the PBE is profitable.
    
    To complete the proof, we verify that $\tau_{\ell}$ and $\tau_{r}$ with the desired properties exist. Since elements of $\mathcal{T}_{\ell}$ and $\mathcal{T}_{r}$ are bounded and since multiplying $\tau_{\ell}$ and $\tau_{r}$ by positive constants preserves \eqref{eq:tau_requirement} and \eqref{eq:tau_profitability}, we may drop the requirement that each $\tau_{i}$ map to $[-1, 1]$.
    \begin{lemma}\label{lemma:tau_construction}
        If the agents' signals are independent, the principal is unbiased, and $v$ is non-separable, then there exist $\tau_{\ell}\in\mathcal{T}_{\ell}$ and $\tau_{r}\in\mathcal{T}_{r}$ that satisfy \eqref{eq:tau_profitability}.
    \end{lemma}

    \begin{proof}[Proof of \Cref{lemma:tau_construction}]
    We show the contrapositive. Suppose $\int_{S} v \tau_{\ell}\tau_{r} \de \pi \leq 0$ for all $\tau_{\ell}\in\mathcal{T}_{\ell}$ and $\tau_{r}\in\mathcal{T}_{r}$. Since $-\tau_{\ell}\in\mathcal{T}_{\ell}$ whenever $\tau_{\ell}\in\mathcal{T}_{\ell}$, it follows that $\int_{S} v \tau_{\ell}\tau_{r} \de \pi = 0$ for all $\tau_{\ell}\in\mathcal{T}_{\ell}$ and $\tau_{r}\in\mathcal{T}_{r}$. We show that $v$ is separable.
    
    For all $i$ and $s_{i}$, let $v_{i}(s_{i}) = \int_{S_{-i}} v(s_{i}, s_{-i}) \de\pi_{-i}(s_{-i}) $.
    If $v_{i}$ is not a.e. equal to $0$, choose bounded $\phi_{i}\colon S_{i}\to \mathbb{R}$ such that $\int_{S_{i}} \phi_{i} \de \pi_{i} = 0$ and $\int_{S_{i}} v_{i} \phi_{i} \de \pi_{i} = 1$.\footnote{For example, using that the principal is unbiased, these properties hold for $\phi_{i}$ defined by: 
    \begin{equation*}
        \forall s_{i},\quad
        \phi_{i}(s_{i}) = \frac{1}{\int_{S_{i}} \vert v_{i}\vert \de\pi_{i}}\left(\sgn(v_{i}(s_{i})) - \int_{S_{i}} \sgn(v_{i})\de \pi_{i}\right).
    \end{equation*}
    }
    Otherwise, let $\phi_{i} = 0$.

To show that $v$ is separable, it suffices to find a separable $\hat{v}$ such that $\int_{S} v \eta_{\ell} \eta_{r} \de \pi = \int_{S} \hat{v} \eta_{\ell} \eta_{r} \de \pi$ for all bounded $\eta_{\ell}\colon S_{\ell}\to\R$ and $\eta_{r}\colon S_{r}\to\R$.\footnote{Setting $\eta_{i} = \mathbbm{1}_{\tilde{S}_{i}}$ shows that $v$ and $\hat{v}$ have the same integral over every measurable rectangle $\tilde{S}_{\ell}\times \tilde{S}_{r}$. Since measurable rectangles form a pi-system that contains $S = S_{\ell}\times S_{r}$ and generates the product sigma-algebra on $S$, it follows from Theorem~16.10(iii) in \citet{billingsley1995probability} that $v = \hat{v}$ holds $\pi$-almost everywhere.}
Since the hypothesis of the contrapositive applies only to elements of $\mathcal{T}_{\ell}$ and $\mathcal{T}_{r}$, we first split each $\eta_{i}$ into an element of $\mathcal{T}_{i}$ and a remainder.
Define the constants $\mu_{i} = \int_{S_{i}} \eta_{i} \de \pi_{i}$ and $k_{i} = \int_{S} v \eta_{i} \de \pi=\int_{S_i} v_i \eta_{i} \de \pi_i$. Set $\tau_{i}(s_i) = \eta_{i}(s_i) - \mu_{i} - k_{i} \phi_{i}(s_i)$.
The choice of $\phi_{i}$ and unbiasedness imply $\int_{S_{i}} \tau_{i} \de \pi_{i} = \int_{S} v \tau_{i} \de \pi = 0$, meaning $\tau_{i} \in \mathcal{T}_{i}$.

Using $\tau_{\ell}\in\mathcal{T}_{\ell}$, $\tau_{r}\in\mathcal{T}_{r}$, and unbiasedness, we have  $\int_{S} v\tau_{\ell}\de\pi = \int_{S} v\tau_{r}\de\pi = \int_{S} v\de\pi = 0$. Moreover, with the hypothesis of the contrapositive, we get $\int_{S} v\tau_{\ell}\tau_{r}\de\pi = 0$. Combining these facts, we have
\begin{align*}
    0
    &= \int_{S} v\tau_{\ell}\tau_{r}\de\pi + \mu_{r}\int_{S} v\tau_{\ell}\de\pi + \mu_{\ell}\int_{S} v\tau_{r}\de\pi + \mu_{\ell}\mu_{r}\int_{S} v\de\pi
    \\
    &= \int_{S} v(\tau_{\ell} + \mu_{\ell})(\tau_{r} + \mu_{r})\de\pi
    \\
    &= \int_{S} v(\eta_{\ell} - k_{\ell}\phi_{\ell})(\eta_{r} - k_{r}\phi_{r})\de\pi,
\end{align*}
where the last equality uses the definition $\tau_{i}(s_i) = \eta_{i}(s_i) - \mu_{i} - k_{i} \phi_{i}(s_i)$.
Expanding the last integral and rearranging yields
\begin{equation*}
    \int_{S} v\eta_{\ell}\eta_{r}\de\pi
    = k_{r}\int_{S} v\eta_{\ell}\phi_{r}\de\pi + k_{\ell}\int_{S} v\phi_{\ell}\eta_{r}\de\pi - k_{\ell}k_{r}K,
\end{equation*}
where  $K = \int_{S} v\phi_{\ell}\phi_{r}\de\pi$ does not depend on $\eta_{\ell}$ or $\eta_{r}$.

    Abbreviate $w_{i}(s_{i}) = \int_{S_{-i}} v(s_{i}, s_{-i}) \phi_{-i}(s_{-i}) \de \pi_{-i}(s_{-i})$ for all $s_{i}$, and note that $w_{\ell}$ and $w_{r}$ do not depend on $\eta_{\ell}$ or $\eta_{r}$.
    The previous expression then simplifies to
    \begin{equation*}
    \int_{S} v\eta_{\ell}\eta_{r}\de\pi
    = k_{r} \int_{S_{\ell}} w_{\ell} \eta_{\ell} \de \pi_{\ell}
       + k_{\ell} \int_{S_{r}} w_{r} \eta_{r} \de\pi_{r} - k_{\ell} k_{r} K.
    \end{equation*}
    Now plugging in the definition $k_{i} = \int_{S} v \eta_{i} \de \pi = \int_{S_{i}} v_{i}\eta_{i} \de \pi_{i}$, we get by the independence of signals and Fubini's theorem
    \begin{align*}
        &\int_{S} v \eta_{\ell} \eta_{r} \de \pi
        \\
        &= 
        \int_{S_{r}} v_{r} \eta_{r} \de\pi_{r} \int_{S_{\ell}} w_{\ell} \eta_{\ell}\de\pi_{\ell} 
        + \int_{S_{\ell}} v_{\ell} \eta_{\ell} \de\pi_{\ell} \int_{S_{r}} w_{r} \eta_{r} \de\pi_{r} 
        - K \int_{S_{\ell}} v_{\ell} \eta_{\ell} \de\pi_{\ell}\int_{S_{r}} v_{r} \eta_{r} \de\pi_{r} 
        \\
        &=
        \int_{S} \left(v_{r} w_{\ell} + v_{\ell} w_{r}  - K v_{\ell} v_{r}\right) \eta_{\ell} \eta_{r} \de \pi=\int_{S} \left(v_{r} w_{\ell} + v_{\ell} \left(w_{r} - K v_{r}\right)\right) \eta_{\ell} \eta_{r} \de \pi.
    \end{align*}
    Thus, $\hat v=v_{r} w_{\ell} + v_{\ell} \left(w_{r} - K v_{r}\right)$ is separable (with integrable components) and does not depend on $\eta_{\ell}$ or $\eta_{r}$. Hence, $v$ is separable, as desired.
    \end{proof}

\subsection{Partial alignment: Proof of \headercref{Proposition}{{proposition:partial_alignment}}}\label{appendix:proof_partial_alignment}

We first show the contrapositive: if no profitable BIC mechanism exists without alignment, then there exists no profitable BIC mechanism under weak alignment. Further ahead, we show by example that the converse fails.

In view of \Cref{proposition:profitability_independence}, since no profitable BIC mechanism exists without alignment, the principal's payoff is additive. That is, there are $v_{\ell}'$ and $v_{r}'$ such that $v(s) = v_{\ell}'(s_{\ell}) + v_{r}'(s_{r})$ for $\pi_{\ell}\otimes\pi_{r}$-almost all $(s_{\ell}, s_{r})$.
Since the principal is unbiased ($\E_{s}[v] = 0)$, we have $\E_{s_{\ell}}[v_{\ell}'] = - \E_{s_{r}}[v_{r}']$. 
Put $v_{\ell} = v_{\ell}' -\E_{s_{\ell}}[v_{\ell}']$ and $v_{r} = v_{r}' - \E_{s_{r}}[v_{r}']$, so that each $v_{i}$ has mean zero and $v= v_{\ell} + v_{r}$.

Now fix $\alpha > 0$ for which partial alignment is weak, i.e., $1 + \alpha \E_{s_{r}}[v_{\ell}(s_{\ell}) + v_{r}(\bm{s_{r}})] > 0 > -1 + \alpha \E_{s_{\ell}}[v_{\ell}(\bm{s_{\ell}}) + v_{r}(s_{r})]$ holds $\pi_{\ell}\otimes\pi_{r}$-almost surely.
Importantly, since each $v_{i}$ has mean zero, also $1 + \alpha v_{\ell}(s_{\ell}) > 0 > -1 + \alpha v_{r}(s_{r})$ holds $\pi_{\ell}\otimes\pi_{r}$-almost surely.

Fix any BIC mechanism $x$. We show $\E_{s}[v(\bm{s}) x(\bm{s})] = 0$. In particular, since the principal is unbiased, this is the payoff from not consulting the agents, so $x$ is not profitable.

As in the proof of \Cref{lemma:BIC_characterization}, we consider deviations in which an agent ignores their signal and reports randomly.
Write $\kappa_{\ell} = 1$ and $\kappa_{r} = -1$, so that agent $i$'s payoff from $L$ is $\kappa_{i} + \alpha v$.
By independence, agent $i$ with signal $s_{i}$ who reports $s_{i}'$ obtains $U_{i}(s_{i}, s_{i}') = \E_{s_{-i}}[(\kappa_{i} + \alpha v_{i}(s_{i}) + \alpha v_{-i}(\bm{s_{-i}})) x(s_{i}', \bm{s_{-i}})]$.
Let $g_{i}\colon S_{i}\to[0, \infty)$ satisfy $\E_{s_{i}}[g_{i}(\bm{s_{i}})] = 1$.
Multiplying the BIC constraint $U_{i}(s_{i}, s_{i}) \geq U_{i}(s_{i}, s_{i}')$ by $g_{i}(s_{i}')$, integrating with respect to $\de\pi_{i}(s_{i}')\de\pi_{i}(s_{i})$, and using $\E_{s_{i}}[g_{i}(\bm{s_{i}})] = 1$ and $\E_{s_{i}}[v_{i}(\bm{s_{i}})] = 0$ yields
\begin{equation*}
    \E_{s}[(\kappa_{i} + \alpha v(\bm{s})) x(\bm{s})] \geq \E_{s}[g_{i}(\bm{s_{i}})(\kappa_{i} + \alpha v_{-i}(\bm{s_{-i}})) x(\bm{s})].
\end{equation*}
For $g_{i} = 1 + \kappa_{i}\alpha v_{i}$, which is strictly positive under weak alignment, the right side equals $\E_{s}[(\kappa_{i} + \alpha v(\bm{s}) + \kappa_{i}\alpha^{2} v_{\ell}(\bm{s_{\ell}}) v_{r}(\bm{s_{r}})) x(\bm{s})]$, so $\kappa_{i}\E_{s}[v_{\ell}(\bm{s_{\ell}}) v_{r}(\bm{s_{r}}) x(\bm{s})] \leq 0$ for both $i$, i.e., $\E_{s}[v_{\ell}(\bm{s_{\ell}}) v_{r}(\bm{s_{r}}) x(\bm{s})] = 0$.
Hence, the inequality displayed above holds with equality for $g_{i} = 1 + \kappa_{i}\alpha v_{i}$.
Since $U_{i}(s_{i}, s_{i}) - U_{i}(s_{i}, s_{i}')$ is non-negative and $g_{i} > 0$, the BIC constraint $U_{i}(s_{i}, s_{i}) - U_{i}(s_{i}, s_{i}')\geq 0$ must also hold with equality for $\pi_{i}\otimes\pi_{i}$-almost all $(s_{i}, s_{i}')$. Consequently, the displayed inequality holds with equality for every $g_{i}$.
For $g_{i} = 1$, this observation yields $\E_{s}[(\kappa_{i} + \alpha v(\bm{s})) x(\bm{s})] = \E_{s}[(\kappa_{i} + \alpha v_{-i}(\bm{s_{-i}})) x(\bm{s})]$, i.e., $\E_{s}[v_{i}(\bm{s_{i}}) x(\bm{s})] = 0$. Since this holds for both agents $i$, we conclude $\E_{s}[v(\bm{s}) x(\bm{s})] = 0$, as desired.

\medskip
We complete the proof by providing an example in which a profitable BIC mechanism exists without alignment but not with weak alignment for a certain $\alpha$.

Let $S_{\ell}=S_r=\{1,2\}$. Signals are independent, with
$\pi_{\ell}=(1/2,1/2)$ and $\pi_r=(2/3,1/3)$. The principal's
payoff from $L$ is
$$
    v=
    \begin{pmatrix}
        -3 & -1\\
         2 &  3
    \end{pmatrix},
$$
where rows correspond to $s_{\ell}$ and columns to $s_r$.
One can verify that the principal is unbiased and that $v$ is non-additive. In particular, a profitable BIC mechanism exists.

Now take $\alpha = 1/4$.
For the chosen $v$, this $\alpha$ results in weak alignment.
Nevertheless, no profitable BIC mechanism $x$ exists.
Let $x$ be BIC.
Three of the BIC
constraints, after multiplying by positive constants, are
\begin{align*}
    2(x(2, 1)-x(1, 1))+3(x(2, 2)-x(1, 2))
        &\leq 0
        &&(IC_\ell:1\to2),\\
    7(x(1, 1)-x(1, 2))+2(x(2, 1)-x(2, 2))
        &\leq 0
        &&(IC_r:1\to2),\\
    5(x(1, 2)-x(1, 1))+(x(2, 2)-x(2, 1))
        &\leq 0
        &&(IC_r:2\to1).
\end{align*}
We sum up these inequalities with respective weights $21$, $46$, and $74$.
This yields:
$$
    90\,\E_s[v(\bm{s}) x(\bm{s})]
    =15(-6x(1, 1)-x(1, 2)+4x(2, 1)+3x(2, 2))
    \leq 0.
$$
In particular, $\E_s[v(\bm{s}) x(\bm{s})] \leq 0$, so that $x$ is not profitable. \qed

\subsection{Multiple Alternatives: Proof of \headercref{Proposition}{{proposition:multiple_alternatives_profitability}}}\label{appendix:proof_multiple_alternatives_profitability}
Let $A_{0} = A\setminus \lbrace a_{0}\rbrace$, where we recall that all payoffs from decision $a_{0}$ are normalized to zero.
We regard a mechanism as a function $x\colon S \to [0, 1]^{\vert A_{0}\vert}$ such that $\sum_{a\in A_{0}} x(a, s) \leq 1$ for all $s\in S$. Likewise, we regard $u$ and $v$, respectively, as being defined on $A_{0}$ and $A_{0}\times S$, respectively.
Agent $\ell$'s payoff at report $s\in S$ is $\sum_{a\in A_{0}} u(a) x(a, s)$, while the principal's payoff is $\sum_{a\in A_{0}} v(a, s) x(a, s)$.

First, suppose there exists $w$ such that $v(a, s) = u(a) w(s)$ for all $a\in A_{0}$ and $s\in S$.
We show that no profitable BIC mechanism exists.
Since $\pi$ has full-rank, $x$ is BIC if and only if $\sum_{a\in A_{0}} u(a) x(a, s)$ is constant in $s$. Let $\bar{u}$ be the constant value.
It follows that the principal's expected utility from $x$ equals $\E_{s}[\sum_{a\in A_{0}} v(a, \bm s) x(a, \bm s)] = \E_{s}[w(\bm{s}) \sum_{a\in A_{0}} u(a) x(a, \bm s)] = \E_{s}[w(\bm{s})] \bar{u} = \sum_{a\in A_{0}} \E_{s}[w(\bm{s})] u(a) \E_{s}[x(a, \bm{s})] = \sum_{a\in A_{0}} \E_{s}[v(a, \bm{s})] \E_{s}[x(a, \bm{s})]$. The final term is simply the principal's expected utility when not consulting the agents and taking each decision $a\in A_{0}$ with probability $\E_{s}[x(a, \bm{s})]$. Thus, no profitable BIC mechanism exists.

Second, suppose there does not exist $w$ such that $v(a, s) = u(a) w(s)$ for all $a\in A_{0}$ and $s\in S$. 
We show that there exists a profitable BIC mechanism $x$.
Let $H$ be the linear space of functions of the form $(a, s)\mapsto u(a) w(s) \pi(s)$ for some $w\colon S\to\R$. 
Since $v$ cannot be written as $v(a, s) = u(a) w(s)$ for all $(a, s)$, and since $\pi$ has full support, we have $v\pi\notin H$ (where $v\pi$ means the function $(a, s)\mapsto v(a, s)\pi(s)$).
Let $h_{0}$ be the orthogonal projection of $v\pi$ onto $H$, and let $\eta = v\pi - h_{0}$ be the projection residual. Thus, $\eta$ is orthogonal to all $h\in H$, i.e., $\sum_{a\in A_{0}, s\in S} \eta(a, s) h(a, s) = 0$.
Let $x_{0}$ be the constant mechanism that takes each decision $a\in A_{0}$ with probability $\tfrac{1}{\vert A\vert}$ (and takes $a_{0}$ with the residual probability).
For $\varepsilon > 0$ sufficiently small, $x = x_{0} + \varepsilon \eta$ is a feasible mechanism.

We claim that $x$ is BIC and profitable.
We first verify BIC. Fix arbitrary $s\in S$.
Notice that the function $(a, s')\mapsto u(a)\mathbbm{1}_{s' = s}$ is in $H$ (here, we are using $\pi(s) > 0$).
Since $\eta$ is orthogonal to $H$, it follows $0 = \sum_{a\in A_{0}, s'\in S} u(a)\mathbbm{1}_{s' = s} \eta(a, s') = \sum_{a\in A_{0}} u(a)\eta(a, s)$.
Therefore, agent $\ell$'s expected payoff at profile $s$ under mechanism $x$ equals $\sum_{a\in A_{0}} u(a) x(a, s) = \tfrac{1}{\vert A\vert}\sum_{a\in A_{0}} u(a)+ \varepsilon\sum_{a\in A_{0}} u(a)\eta(a, s) = \tfrac{1}{\vert A\vert}\sum_{a\in A_{0}} u(a)$.
This payoff does not depend on $s$. 
Agent $r$'s payoff is simply the negative of this payoff.
Thus, $x$ is BIC.

We now verify that $x$ is profitable.
Since the principal is ex-ante indifferent, it holds
\begin{equation*}
 \sum_{a\in A_{0}, s\in S} x_{0}(a, s) v(a, s) \pi(s) = \tfrac{1}{\vert A\vert}\sum_{a\in A_{0}, s\in S} v(a, s) \pi(s)= 0.   
\end{equation*}
Thus, we show 
$    \sum_{a\in A_{0}, s\in S} \eta(a, s) v(a, s) \pi(s) > 0.
$
The orthogonal projection $h_{0}$ is in $H$.
Since $\eta$ is orthogonal to $H$, we get 
\begin{align*}
    \sum_{a\in A_{0}, s\in S} \eta(a, s)v(a, s)\pi(s) &= \sum_{a\in A_{0}, s\in S} \eta(a, s) (\pi(s) v(a, s) - h_{0}(a, s)) \\ &= \sum_{a\in A_{0}, s\in S} \eta(a, s)^{2}.
\end{align*}
Since $v\pi\notin H$, the projection residual $\eta$ is non-zero. 
Thus, $\sum_{a\in A_{0}, s\in S} \eta(a, s)^{2} > 0$.
\qed

\addcontentsline{toc}{section}{References}
\newrefcontext[sorting=nyt]
\printbibliography

\clearpage

\begin{center}
    {\Large\scshape Online Appendix}  
\end{center}

\appendix
\setcounter{section}{0}
\renewcommand{\thesection}{O.\arabic{section}}
\renewcommand{\theHsection}{O.\arabic{section}}

\section{Optimal Mechanisms with Discrete Signals and Randomization}\label{appendix:discrete_randomization}

In this section, we discuss optimal mechanisms for finite signal spaces and, in contrast to the non-atomic setting underlying \Cref{theorem:optimal_mech}, show that randomization is typically necessary for optimality (unless no profitable mechanism exists). We present three results.
First, we show that non-constant optimal BIC mechanisms necessarily involve randomization for generic marginal signal distributions $\pi_\ell$ or $\pi_r$
(\Cref{proposition:discrete_generic_stochastic}).
Second, we show that there is nevertheless an upper bound on the share of signal realizations at which some optimal mechanism must randomize, and this bound converges to zero as signal spaces become rich (\Cref{proposition:discrete_randomization_bound}). 
Third, we give a complete characterization of the extreme points of the set of BIC mechanisms (\Cref{prop:discrete_extreme_points}).
If $\pi$ has full support on a finite signal space, then every extreme point is the  uniquely optimal mechanism for some objective $v$ of the principal (because every extreme point of a polytope is exposed); hence the extreme points are the relevant candidates for optimality.

For the first result, fix finite signal spaces $S_{\ell}$ and $S_{r}$.
Recall that $\mathcal{X}(\pi)$ denotes the set of BIC mechanisms given joint signal distribution $\pi$.
Here, a mechanism is a $\vert S_{\ell}\vert\times\vert S_{r}\vert$-matrix.
Let $\vec{0}$ and $\vec{1}$, respectively, denote the $\vert S_{\ell}\vert\times\vert S_{r}\vert$ matrices of all zeros and ones, respectively, corresponding to the two constant deterministic mechanisms.
Let $\mathcal{D}(\pi)$ denote the subset of deterministic BIC mechanisms under $\pi$. 
It is clear that the two constant mechanisms, i.e., $x=\vec{0}$ and $x=\vec{1}$, are in $\mathcal{D}(\pi)$. The following proposition asserts that for generic marginal signal distributions, these two mechanisms are the only deterministic BIC mechanisms.
\begin{proposition}\label{proposition:discrete_generic_stochastic}
	There exist sets of marginal signal distributions $\mathcal{H}_\ell\subset \Delta(S_{\ell})$ and $\mathcal{H}_r\subset \Delta(S_{r})$, each of which is a union of finitely many hyperplanes restricted to the respective simplex, such that if $\mathcal{D}(\pi)\neq\{\vec{0},\vec{1}\}$, then $(\pi_\ell,\pi_r)\in \mathcal{H}_\ell\times \mathcal{H}_r$.
\end{proposition}

There is little hope of circumventing this result by focusing on economically important, non-generic environments because the hyperplanes in \Cref{proposition:discrete_generic_stochastic} are chosen separately for each marginal. 
That is, the conclusion holds whenever one of the marginals is generic, for instance, for independent signals with generic marginals.

\begin{proof}[Proof of \Cref{proposition:discrete_generic_stochastic}]
	We first construct $\mathcal{H}_\ell\subset \Delta(S_{\ell})$ and $\mathcal{H}_r\subset \Delta(S_{r})$ and then demonstrate the desired properties.

	Fix arbitrary $\tilde{S}_\ell,\bar{S}_\ell\subset S_\ell$ such that $\tilde{S}_{\ell}\neq \bar{S}_{\ell}$. Let $H(\tilde{S}_\ell,\bar{S}_\ell)\subset \mathbb{R}^{|S_{\ell}|}$ denote the hyperplane associated with the following linear equation
	\begin{equation*}
		\label{eq:hyperplane}
		\sum_{s_{\ell}}\left(\mathbbm{1}_{\tilde{S}_\ell}(s_{\ell})-\mathbbm{1}_{\bar{S}_\ell}(s_{\ell})\right) \beta(s_\ell)=0.
	\end{equation*}
	If $H(\tilde{S}_\ell,\bar{S}_\ell)\cap \Delta(S_{\ell})$ is non-empty, then $H(\tilde{S}_\ell,\bar{S}_\ell)$ has codimension 1 in $\Delta(S_{\ell})$ because $\Delta(S_{\ell})$ is a hyperplane (restricted to the positive orthant) that cannot coincide with $H(\tilde{S}_\ell,\bar{S}_\ell)$. Hence, $$\mathcal{H}_\ell=\bigcup_{\substack{\tilde{S}_\ell,\bar{S}_\ell\subset S_\ell\\\tilde{S}_{\ell}\neq \bar{S}_{\ell}}} H(\tilde{S}_\ell,\bar{S}_\ell)\cap \Delta(S_{\ell})$$ is a finite union of hyperplanes restricted to $\Delta(S_{\ell})$. The construction of $\mathcal{H}_r\subset \Delta(S_{r})$ is analogous.

	Now fix an arbitrary deterministic mechanism $x\notin\{\vec{0},\vec{1}\}$. We shall complete the proof by showing that if $\pi\in\Pi$ is such that, say, the marginal $\pi_\ell$ is not in $\mathcal{H}_\ell$, then $x$ is not BIC. If $x$ depends only on $s_\ell$, then $x$ violates BIC for agent $\ell$. Otherwise, there exist $\tilde{s}_r,\bar{s}_r\in S_r$ such that $x(\cdot,\tilde{s}_r)\neq x(\cdot,\bar{s}_r)$. Let $\tilde{S}_{\ell}=\{s_\ell\in S_{\ell}\mid x(s_\ell,\tilde{s}_r)=1\}$ and $\bar{S}_{\ell}=\{s_\ell\in S_{\ell}\mid x(s_\ell,\bar{s}_r)=1\}$. By construction, $\tilde{S}_{\ell}\neq \bar{S}_{\ell}$. By assumption, $\pi_\ell\notin\mathcal{H}_\ell$ and, a fortiori, $\pi_\ell\notin H(\tilde{S}_\ell,\bar{S}_\ell)$, so
	$$\sum_{s_{\ell}}\left(x(s_{\ell},\tilde{s}_r)-x(s_{\ell},\bar{s}_r)\right) \pi_\ell(s_\ell)\neq 0.$$
	Therefore, $x$ is not BIC under the independent signal distribution $\pi_{\ell}\otimes\pi_r$, i.e., $x\notin\mathcal{X}(\pi_{\ell}\otimes\pi_r)$. By \Cref{proposition:spanning_order}, $x$ is not BIC under $\pi$ either, i.e., $x\notin\mathcal{X}(\pi)\subseteq\mathcal{X}(\pi_{\ell}\otimes\pi_r)$.
\end{proof}

Our second result bounds the number of signal profiles at which an extreme point of the set of BIC mechanisms can randomize when signals are independent.
Since the principal's expected utility is linear in the mechanism, there always exists an optimal BIC mechanism that is an extreme point.

\begin{proposition}\label{proposition:discrete_randomization_bound}
Suppose signals are independent. If $x$ is an extreme point of $\mathcal{X}(\pi)$, then $\vert\lbrace s\in S\colon x(s)\in(0,1)\rbrace\vert\leq |S_\ell|+|S_r|-2$.
\end{proposition}
In particular, the share of signal profiles at which the mechanism randomizes, which is at most $(|S_\ell|+|S_r| - 2)/(|S_\ell|\times|S_r|)$, vanishes as $\min\{|S_\ell|,|S_r|\}\to\infty$. 

Moreover, the bound in \Cref{proposition:discrete_randomization_bound} is tight. For example, for $m:=|S_\ell|=|S_r|\geq 4$, the following matrix defines an extreme point of the set of BIC mechanisms when signals are uniform (which can be verified by applying \Cref{prop:discrete_extreme_points} below):
\begin{equation*}
x=
\begin{pmatrix}
    1 & 0 & \dots  & 0 & \dfrac{1}{m-2}\\
    0 & 1 & \dots  & 0 & \dfrac{1}{m-2}\\
    \vdots & \vdots & \ddots & \vdots & \vdots \\[2pt]
    0 & 0 & \dots  & 1 & \dfrac{1}{m-2}\\[4pt]
    \dfrac{1}{m-2} & \dfrac{1}{m-2} & \dots & \dfrac{1}{m-2} & 0
\end{pmatrix}
.
\end{equation*}

We prove \Cref{proposition:discrete_randomization_bound}
using the following extreme-point characterization.
Let $G$ denote the (undirected) complete bipartite graph with vertex set $S_\ell\sqcup S_r$. For a mechanism $x\in\mathcal{X}(\pi)$, let $E(x)$ denote the subgraph of $G$ with vertex set $S_\ell\sqcup S_r$ and edges $s_\ell s_r$ if and only if $x(s_\ell,s_r)\in(0,1)$. Say that a subgraph of $G$ with vertices $V_\ell\subseteq S_\ell$ and $V_r\subseteq S_r$ is \textit{balanced} if $\pi_\ell(V_\ell)=\pi_r(V_r)$ and \textit{unbalanced} otherwise. Recall that a \textit{tree} is a connected acyclic graph, and a \textit{forest} is a disjoint union of trees. In particular, an isolated vertex is an unbalanced tree because we have assumed without loss that the marginals $\pi_\ell$ and $\pi_r$ have full support. 

\begin{proposition}\label{prop:discrete_extreme_points}
	Suppose signals are independent.
    Then, $x\in\mathcal{X}(\pi)$ is an extreme point of $\mathcal{X}(\pi)$ if and only if $E(x)$ is a forest such that at least one tree in $E(x)$ is unbalanced.
\end{proposition}
We first prove \Cref{proposition:discrete_randomization_bound} and then \Cref{prop:discrete_extreme_points}.

\begin{proof}[Proof of \Cref{proposition:discrete_randomization_bound}]
	By \Cref{prop:discrete_extreme_points}, $E(x)$ is a forest with at least one unbalanced tree. A single tree containing all vertices is balanced since $\pi_\ell(S_\ell)=\pi_r(S_r)=1$. Hence, $E(x)$ consists of at least two trees. Since a forest with $K$ trees on $|S_\ell|+|S_r|$ vertices has $|S_\ell|+|S_r|-K$ edges, $E(x)$ has at most $|S_\ell|+|S_r|-2$ edges. Thus, $x$ randomizes at most at $|S_\ell|+|S_r|-2$ signal profiles.
\end{proof}

\begin{proof}[Proof of \Cref{prop:discrete_extreme_points}]
	Recall that, with independent signals, a mechanism $x'$ is BIC if and only if $\sum_{s_r}\pi_r(s_r)x'(s_\ell,s_r)$ is constant in $s_\ell$ and $\sum_{s_\ell}\pi_\ell(s_\ell)x'(s_\ell,s_r)$ is constant in $s_r$, where both constants equal $\sum_{s}\pi_\ell(s_\ell)\pi_r(s_r)x'(s)$ (\Cref{lemma:BIC_characterization}).

	We first observe that if $x$ is an extreme point of $\mathcal{X}(\pi)$, then $E(x)$ must be a forest.\footnote{The set of BIC mechanisms with a \textit{fixed} ex-ante decision probability $\bar x$ forms, after the change of variables $z=x\odot(\pi_\ell\otimes\pi_r)$, a capacity-constrained transportation polytope. A point of such a polytope is extreme if and only if the entries for which the capacity constraints $0\leq z\leq \pi_\ell\otimes\pi_r$ are both slack form a forest (cf.\ \citet[Def.~9.21 and Cor.~9.23]{korte2000combinatorial} for the more general case of flow polytopes). To our knowledge, the endogenous scalar $\bar x$ has no counterpart in the combinatorial optimization literature and is what requires our notion of an unbalanced tree.}
    Since $G$ is bipartite, if $E(x)$ would contain a cycle, the cycle would have to be even. By alternating $+\varepsilon/(\pi_\ell(s_\ell)\pi_r(s_r))$ and $-\varepsilon/(\pi_\ell(s_\ell)\pi_r(s_r))$ clockwise and counterclockwise on the edges $s_\ell s_r$ of this cycle, we obtain a convex decomposition of $x$. 

	Thus, we must show that if $x\in\mathcal{X}(\pi)$ and $E(x)$ is a forest, then $x$ is an extreme point of $\mathcal{X}(\pi)$ if and only if at least one tree of $E(x)$ is unbalanced.

	Consider the homogeneous linear system
	\begin{align*}
		\sum_{s_r:\: s_\ell s_r\in E(x)} \pi_r(s_r)d(s_\ell,s_r)&=\alpha \quad\forall s_\ell\\
		\sum_{s_\ell:\: s_\ell s_r\in E(x)} \pi_\ell(s_\ell)d(s_\ell,s_r)&=\alpha \quad\forall s_r
	\end{align*}
	with one variable $d(s_\ell,s_r)$ for each edge of $E(x)$ and another variable $\alpha\in\mathbb{R}$. $x$ is not an extreme point of $\mathcal{X}(\pi)$ if and only if this system has a solution other than the trivial one $(\vec{0},0)$. Substituting $z(s_\ell,s_r)=\pi_\ell(s_\ell)\pi_r(s_r)d(s_\ell,s_r)$, the system becomes $Az=\alpha w$, where $A$ is the vertex-edge incidence matrix of $E(x)$ and $w$ is the vector that assigns $\pi_\ell(s_\ell)$ to every vertex $s_\ell$ and $\pi_r(s_r)$ to every vertex $s_r$.

    Since $E(x)$ is a forest, up to reordering  the rows and columns of $A$, $A$ is block-diagonal with blocks $A_k$ corresponding to the trees of $E(x)$ (i.e., its connected components). Since the columns of the incidence matrix $A_k$ of a tree are linearly independent and $A$ is block diagonal, the columns of $A$ are also linearly independent.

    Since the columns of $A$ are linearly independent, a non-trivial solution to $Az=\alpha w$ exists if and only if $w\in\col(A)$. Equivalently, $w_k\in\col(A_k)$ for every $k$, where $w_k$ is the restriction of $w$ to the vertices of the $k$-th tree of $E(x)$. Thus, it remains to show that $w_k\in\col(A_k)$ if and only if the $k$-th tree is balanced.

		To prove the remaining claim, let $V_k$ and $E_k$ denote the vertex and edge sets of the $k$-th tree. Note that $y_k^\top A_k=0$, where $y_k^\top$ is the row vector that assigns $+1$ to every vertex in $V_k\cap S_\ell$ and $-1$ to every vertex in $V_k\cap S_r$, because $G$ is bipartite. The columns of $A_k$ are linearly independent, hence $\rank(A_k)= |E_k|=|V_k|-1$. By the rank nullity theorem, the co-kernel of $A_k$ has dimension 1 and is hence the span of $y_k$. Since the co-kernel is orthogonal to $\col(A_k)$, we have $w_k\in\col(A_k)$ if and only if $y_k^\top w_k=0$, which means exactly that the $k$-th tree is balanced.
\end{proof}

\section{Example: Perfect Complements}\label{appendix:CES_production}

A decision-maker decides whether to adopt a project (decision $L$) whose output is limited by the scarcer of two complementary inputs.
Agent $\ell$ advocates for adoption and privately knows the capacity $s_{\ell} \in [0, 1]$ of one input.
Agent $r$ argues against adoption and privately knows the capacity $s_{r}\in [0, 1]$ of the other input.
The decision-maker's payoff from adoption is the output minus a cost $c$:
\begin{equation*}
    v(s_{\ell}, s_{r}) = \min\lbrace s_{\ell}, s_{r}\rbrace - c.
\end{equation*}

For example, a regulator decides whether a network operator, such as the owner of a power grid, a pipeline, or a rail track, must give a rival access to its network, where the rival knows how much it would send to the network and the operator, which would face competition from the rival, knows how much spare capacity its network has.
As another example, suppose management decides whether to deploy an AI system, where the vendor knows the share of tasks its system can perform and the works council, which represents the employees the system would replace, knows how many tasks the remaining staff can supervise.

Let $s_{\ell}$ and $s_{r}$ be independently and uniformly distributed on $[0, 1]$, and suppose $c$ is such that the decision-maker is ex-ante indifferent, $\mathbb{E}_{s}\left[v(\bm{s})\right] = 0$, i.e., $c = 1/3$.

In this example, $\tilde{v}(s_{\ell}, s_{r}) = \min\lbrace s_{\ell}, s_{r}\rbrace - s_{\ell} - s_{r} + (s_{\ell}^{2} + s_{r}^{2})/2 + 1/3$. 
The guess-and-verify approach does not work as $\tilde{x} = \mathbbm{1}(\tilde{v} > 0)$ is not BIC. 
Indeed, as shown in \Cref{fig:minimum_mech}, under $\tilde{x}$ agent $\ell$'s expected decision probability $\mathbb{E}_{s_{r}}[\tilde{x}(s_{\ell}, \bm{s_{r}})]$ is strictly decreasing for $s_{\ell} \in [3/4, 1]$, violating BIC.

We describe a virtual value $v^{\ast}$ and explain how it defines an optimal BIC mechanism $x^{\ast}$. 
For a parameter $\alpha \in (0, 1/2)$ to be chosen, for all $s_{\ell}, s_{r}$ let
\begin{equation*}
    v^{\ast}(s_{\ell}, s_{r}) = \min\lbrace s_{\ell}, s_{r}\rbrace + \alpha - h(s_{\ell}) - h(s_{r}),
\end{equation*}
where $h(t) = \min\lbrace t, \alpha\rbrace + \min\lbrace t, 1-\alpha\rbrace - \min\lbrace t, 1/2\rbrace$ for all $t$.
For $v^{\ast}$ to be BIC-equivalent to $v$, the difference $v^{\ast} - v$ must have mean zero, which requires $2\alpha^{2} - \alpha + 1/12 = 0$. We take the larger root, $\alpha = (3 + \sqrt{3})/12 \approx 0.394$ (only for this root does the construction below yield a BIC mechanism).
\Cref{fig:minimum_virtual_value} depicts the sign of $v^{\ast}$.
Our candidate optimal BIC mechanism $x^{\ast}$ is then to take decision $L$ with probability one (resp. zero) whenever $v^{\ast}$ is strictly positive (resp. strictly negative).
On the set where $v^{\ast}$ equals zero, exactly one agent's signal $s_{i}$ lies in $[\alpha, 1-\alpha]$, and $x^{\ast}$ takes decision $L$ with probability
\begin{equation*}
    \frac{5/8 + \sqrt{3}/12 - \min\lbrace s_{i}, 1 - s_{i}\rbrace}{\alpha} \in (0, 1).
\end{equation*}
One may check that each agent then expects decision $L$ with probability $5/8 - \sqrt{3}/12 \approx 0.481$ for almost all reports.
The only exceptional report is $1/2$. 
So, we modify $x^{\ast}$ so that if any agent reports $1/2$, then decision $L$ is taken with probability $5/8 - \sqrt{3}/12$.
(This does not upset the expectations from any other report as the other agent's reporting $1/2$ is a measure-zero event.)
Thus, $x^{\ast}$ is BIC.
In words, the mechanism $x^{\ast}$ is as follows:
\begin{quote}
\emph{each agent $i$ separately reports a capacity $\hat{s}_{i} \in [0, 1]$ to the decision-maker; the decision-maker adopts if both reports are low ($\hat{s}_{i} < \alpha$), both are high ($\hat{s}_{i} > 1 - \alpha$), or both are intermediate and differ by less than $1/2 - \alpha$; the decision-maker randomizes if one report is low or high and the other is intermediate and on the same side of $1/2$ (or one report is $1/2$); and the decision-maker rejects otherwise.}
\end{quote}
To verify that $x^{\ast}$ is optimal, take an arbitrary BIC mechanism $x$. Then:
\begin{equation*}
    \mathbb{E}_{s}[v(\bm{s}) x(\bm{s})] = \mathbb{E}_{s}[v^{\ast}(\bm{s}) x(\bm{s})] \leq \mathbb{E}_{s}[v^{\ast}(\bm{s}) x^{\ast}(\bm{s})] = \mathbb{E}_{s}[v(\bm{s}) x^{\ast}(\bm{s})],
\end{equation*}
where the two equalities follow from BIC-equivalence of $v^{\ast}$ and $v$, while the inequality follows from the construction of $x^{\ast}$.

\begin{figure}[ht]
\centering
\tikzset{
    pos region/.style={fill=gray!70},
    zero region/.style={fill=gray!20, postaction={pattern=north east lines, pattern color=gray}},
    v zero/.style={dotted, very thick},
    tick label/.style={font=\scriptsize, text height=1.5ex, text depth=.25ex},
}
\begin{subfigure}[t]{0.48\textwidth}
    \centering
    \begin{tikzpicture}[scale=6]
        \pgfmathsetmacro{\rad}{1/sqrt(3)}
        \fill[pos region] (0,0) rectangle (1,1);
        \fill[white] ({1-\rad},0) arc[start angle=180, end angle=90, radius=\rad] -- (1,0) -- cycle;
        \fill[white] (\rad,1) arc[start angle=0, end angle=-90, radius=\rad] -- (0,1) -- cycle;
        \draw[v zero] (1/3,1) -- (1/3,1/3) -- (1,1/3);
        \draw[thick] (0,0) rectangle (1,1);
        \foreach \t/\lab in {0/0, 0.5/1\slash 2, 1/1} {
            \draw (\t,0) -- (\t,-0.015) node[below, tick label] {$\lab$};
            \draw (0,\t) -- (-0.015,\t) node[left, tick label] {$\lab$};
        }
        \node[below=18pt] at (0.5,0) {$s_{\ell}$};
        \node[left=22pt] at (0,0.5) {$s_r$};
        \begin{scope}[shift={(0.22,-0.27)}, tick label]
            \fill[pos region] (0,0) rectangle (0.08,0.04);
            \node[right] at (0.09,0.02) {$\tilde{x} = 1$};
            \draw[v zero] (0.32,0.02) -- (0.42,0.02);
            \node[right] at (0.43,0.02) {$v = 0$};
            \draw[gray] (0,-0.08) rectangle (0.08,-0.04);
            \node[right] at (0.09,-0.06) {$\tilde{x} = 0$};
        \end{scope}
    \end{tikzpicture}
    \caption{}
    \label{fig:minimum_mech}
\end{subfigure}
\hfill
\begin{subfigure}[t]{0.48\textwidth}
    \centering
    \begin{tikzpicture}[scale=6]
        \pgfmathsetmacro{\al}{(3+sqrt(3))/12}
        \fill[pos region] (0,0) rectangle (\al,\al);
        \fill[pos region] ({1-\al},{1-\al}) rectangle (1,1);
        \fill[pos region] (\al,\al) -- (0.5,\al) -- ({1-\al},0.5) -- ({1-\al},{1-\al}) -- (0.5,{1-\al}) -- (\al,0.5) -- cycle;
        \fill[zero region] (\al,0) rectangle (0.5,\al);
        \fill[zero region] (0,\al) rectangle (\al,0.5);
        \fill[zero region] (0.5,{1-\al}) rectangle ({1-\al},1);
        \fill[zero region] ({1-\al},0.5) rectangle (1,{1-\al});
        \draw[v zero] (1/3,1) -- (1/3,1/3) -- (1,1/3);
        \draw[thick] (0,0) rectangle (1,1);
        \foreach \t/\lab/\dx in {0/0/0, \al/\alpha/-2, 0.5/1\slash 2/-2, 1-\al/1-\alpha/4, 1/1/0} {
            \draw (\t,0) -- (\t,-0.015) node[below, tick label, xshift=\dx pt] {$\lab$};
            \draw (0,\t) -- (-0.015,\t) node[left, tick label] {$\lab$};
        }
        \node[below=18pt] at (0.5,0) {$s_{\ell}$};
        \node[left=30pt] at (0,0.5) {$s_r$};
        \begin{scope}[shift={(0.22,-0.27)}, tick label]
            \fill[pos region] (0,0) rectangle (0.08,0.04);
            \node[right] at (0.09,0.02) {$v^{\ast} > 0$};
            \fill[zero region] (0.32,0) rectangle (0.40,0.04);
            \node[right] at (0.41,0.02) {$v^{\ast} = 0$};
            \draw[gray] (0,-0.08) rectangle (0.08,-0.04);
            \node[right] at (0.09,-0.06) {$v^{\ast} < 0$};
            \draw[v zero] (0.32,-0.06) -- (0.42,-0.06);
            \node[right] at (0.43,-0.06) {$v = 0$};
        \end{scope}
    \end{tikzpicture}
    \caption{}
    \label{fig:minimum_virtual_value}
\end{subfigure}
\caption{In the example of \Cref{appendix:CES_production}, the mechanism $\tilde{x}$ (left panel) is not BIC. The right panel shows the sign of a virtual value $v^{\ast}$. An optimal BIC mechanism $x^{\ast}$ thus takes decision $L$ with probability $1$ on the dark region, with the probability given in the text in the lightly shaded region, and with probability $0$ in the white region. The first-best decision is to take $L$ at all points to the north-east of the dotted line.}
\label{fig:minimum}
\end{figure}

\end{document}